\documentclass[10pt,journal]{IEEEtran}

\usepackage{amsmath}
\usepackage{amssymb}
\usepackage{graphicx}
\usepackage{cite}
\usepackage{mathrsfs}
\usepackage{stmaryrd}
\usepackage{hyperref}
\usepackage{pdfsync}
\usepackage{subfigure}
\usepackage{color}
\usepackage{booktabs}
\usepackage{bbm}

\newtheorem{theorem}{Theorem}

\newtheorem{lemma}{Lemma}

\newtheorem{definition}{Definition}

\newtheorem{example}{Example}

\newcommand{\mbf}[1]{\ensuremath{\mathbf{#1}}}
\newcommand{\bm}[1]{\boldsymbol{#1}}
\newcommand{\set}[1]{\ensuremath{\mathcal{#1}}} % Alphabets and sets.
\newcommand{\bs}[1]{^{\hspace{-0.2mm}\scriptscriptstyle (\hspace{-0.2mm}#1\hspace{-0.2mm})}} \newcommand{\Cu}[1]{C\bs{#1}_{\text{up}}}

\newcommand{\markov}{\textnormal{\mbox{$\multimap\hspace{-0.73ex}-\hspace{-2ex}-$}}}
\newcommand{\indicator}[1]{\mathbbm{1}{\left\{ {#1} \right\} }}
\newcommand{\step}[2]{\stackrel{(\text{#1})}{#2}}

\newcommand{\pchk}{\mathtt{H}}

\allowdisplaybreaks[4] 

\title{\huge Multi-Way Relay Networks: Orthogonal Uplink, Source-Channel Separation and Code Design}

\author{Roy Timo, Gottfried Lechner, Lawrence Ong and Sarah J.\ Johnson
\thanks{R. Timo and G. Lechner are with the Institute for Telecommunications Research at the University of South Australia (e-mail: roy.timo@unisa.edu.au; rtimo@princeton.edu; gottfried.lechner@unisa.edu.au).}
\thanks{L. Ong and S. Johnson are with the School of Electrical Engineering and Computer Science, University of Newcastle (e-mail: lawrence.ong@cantab.net; sarah.johnson@newcastle.edu.au).}
\thanks{R. Timo, G. Lechner and S. Johnson are supported by the Australian Research Council Discovery Grant DP120102123.}
\thanks{L. Ong is supported by the Australian Research Council Discovery Early Career Researcher Award DE120100246.}
\thanks{Parts of this paper were presented at the Data Compression Conference, Snowbird, UT, March, 2011.}
}

\begin{document}

\maketitle

\begin{abstract}
We consider a multi-way relay network with an orthogonal uplink and correlated sources, and we characterise reliable communication (in the usual Shannon sense) with a single-letter expression. The characterisation is obtained using a joint source-channel random-coding argument, which is based on a combination of Wyner \emph{et al.'s}  \emph{Cascaded Slepian-Wolf Source Coding} and Tuncel's \emph{Slepian-Wolf Coding over Broadcast Channels.} We prove a separation theorem for the special case of two nodes; that is, we show that a modular code architecture with separate source and channel coding functions is (asymptotically) optimal. Finally, we propose a practical coding scheme based on low-density parity-check codes, and we analyse its performance using multi-edge density evolution.
\end{abstract}

%\begin{keywords}
%Joint source-channel coding, two-way relay network, source-channel separation.
%\end{keywords}

%%%%
%%%% Section 1: Introduction
%%%%

\section{Introduction}

\IEEEPARstart{C}{onsider} a multi-way relay network in which a group of physically separated nodes exchange data. Direct communication between the nodes is not permitted, and the exchange is only made possible with the help of a relay. The nodes encode and transmit their data over an \emph{uplink} (a multiple-access channel) to the Relay. The Relay processes this information and transmits over the \emph{downlink} (a broadcast channel) to every node. We assume that each node requires a lossless reconstruction of the data of all other nodes.

The above relay network aims to model communication in cellular and satellite networks. A large body of work has comprehensively studied the network from the perspective of source coding~\cite{Wyner-Jun-2002-A}, channel capacity~\cite{Rankov-Jul-2006-C,Knopp-Feb-2006-C,Rankov-Feb-2007-A,Nam-Nov-2010-A,Cui-2012-A}, and network coding~\cite{Cui-Oct-2009-A,Katti-Jun-2008-A}. However, despite this intense effort, the information-theoretic limits of the network remain largely unknown. 

We study the relay network under two specific assumptions. The first assumption is that the data is arbitrarily correlated -- generated by a discrete-memoryless (DM) source -- and the communications problem involves joint source-channel (JSC) coding. Correlated data might take the form of measurements in a sensor network~\cite{Jindal-Nov-2008-A}, voice data in a cellular network, and data files in a peer-to-peer network. We wish to determine when a given source can be reliably communicated (in the usual Shannon sense) over a given channel. 

\begin{figure*}
\begin{center}
\includegraphics[width=1.8\columnwidth]{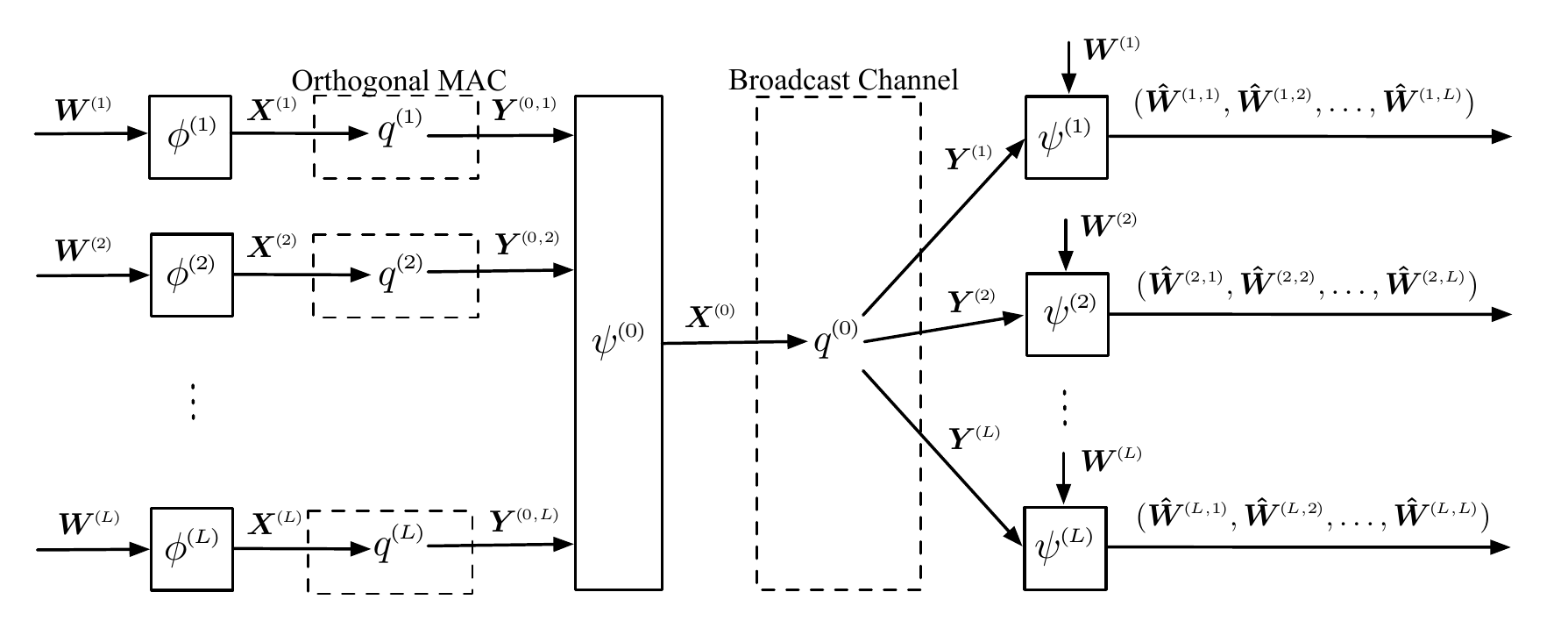}
\caption{The JSC-code architecture for the multi-way relay network with an orthogonal uplink.}
\label{Fig:JSCCArc}
\end{center}
\end{figure*}

The second assumption is that the \emph{downlink} is an arbitrary DM broadcast channel and the \emph{uplink} is an orthogonal DM multiple-access channel. Our motivation to study an orthogonal uplink stems from Shannon's classic separation theorem~\cite[Sec. 7.3]{Cover-2006-B}, which states that the problem of losslessly transmitting a DM source over a point-to-point DM channel can be divided into two independent problems -- source coding and channel coding. Moreover, the individual optimisation of stand-alone source and channel codes is optimal for the overall point-to-point JSC problem. The separation theorem is important in practice because, for example, systems are rarely restricted to transmitting a single source over a fixed channel with known statistics; indeed, to quote Gallager~\cite[Pg. 140]{Gallager-1968-B} (see also~\cite[Pg.~406]{Ahlswede-May-1983-A}):
\begin{quote}
 \emph{``In many data transmission systems the probabilities with which the messages are to be used are either unknown or unmeaningful."}
\end{quote}
The modular nature of the separate source-channel coding architecture allows the source and channel codes to be changed as needed, without compromising overall optimality~\cite{Yeung-2008-B}. Unfortunately, separation may or may not be optimal for networks in general; for example, separation is suboptimal for the multiple-access channel~\cite[Pg.~592]{Cover-2006-B} and the broadcast channel~\cite[Sec.~14.2]{El-Gamal-2011-B}, and it is optimal for the orthogonal multiple-access channel~\cite{Han-1980-A}. Given this state of affairs, it is natural to ask whether separation is optimal for the multi-way relay network. We prove, in this paper, that separation holds for the special case of two nodes. 

\medskip

\emph{Paper Outline:}
\begin{itemize}
\item {\bf Section~\ref{Sec:JSCC}:} We formally define the JSC-coding multi-way relay problem, and we characterise reliable communication with matching single-letter achievability and converse theorems. The achievability proof employs a JSC random-coding argument, which builds on the \emph{virtual binning} idea of Tuncel~\cite{Tuncel-Apr-2006-A} and the \emph{cascaded Slepian-Wolf binning} idea of Wyner \emph{et al.}~\cite{Wyner-Jun-2002-A}.
\item {\bf Section~\ref{Sec:SSCC}:} We formalise a notion of source-channel separation, and we prove a separation theorem for \emph{two} nodes; that is, it is asymptotically (in blocklength) optimal to separate the source and channel coding functions. 
\item {\bf Section~\ref{Sec:Practical}:} We use the two-node separation theorem as a basis to design practical low complexity codes. Specifically, we consider source and channel codes based on low-density parity-check (LDPC) codes. We show how the individual LDPC codes for source and channel coding can be represented by a joint factor-graph~\cite{Kschischang-Feb-2001-A}, and we use this graph to provide an alternative view of the separation theorem. Finally, we present a numerical example and discuss the differences between joint and separate decoding.
\end{itemize}

\medskip

\emph{Notation:} Random variables and their alphabets are identified by uppercase and script letters respectively, e.g. $W$ and $\set{W}$. Random vectors defined on the cartesian product of a set are identified by boldface font, e.g.,
\begin{equation*}
\bm{W} = (W_1,W_2,\ldots,W_n)
\end{equation*}
takes values from
\begin{equation*}
\bm{\set{W}} = \underbrace{\set{W} \times \set{W} \times \cdots \times \set{W}}_{n}.
\end{equation*}
Subsets and strict subsets of an alphabet are identified with $\subseteq$ and $\subset$ respectively. Set complement is denoted by a superscript $c$; e.g., if
\begin{equation*}
\set{L} \subseteq \{1,2,\ldots,L\},
\end{equation*}
then
\begin{equation*}
\set{L}^c \triangleq \{1,2,\ldots,L\} \backslash \set{L}.
\end{equation*}
If $\set{L}$ is a singleton, say $\set{L} = \{l\}$, then we write $\sim l \triangleq \{l\}^c$.

%%%%
%%%% Section 2: Reliable Communication
%%%%

\section{Joint Source-Channel Coding}\label{Sec:JSCC}

\subsection{Setup}

Consider Fig.~\ref{Fig:JSCCArc}. Suppose that a \emph{discrete memoryless source} emits an i.i.d. string
\begin{multline}\label{Eqn:DMS}
(W\bs{1}_1,W\bs{2}_1,\ldots,W\bs{L}_1), (W\bs{1}_2,W\bs{2}_2,\ldots,W\bs{L}_2),\ldots\\
 (W\bs{1}_n,W\bs{2}_n,\ldots,W\bs{L}_n),
\end{multline}
of arbitrarily distributed random variables $(W\bs{1},W\bs{2},\ldots,$ $W\bs{L})$. Let $\set{W}\bs{l}$ denote the alphabet of the $l$-th random variable $W\bs{l}$ for each index $l$ in $\{1,2,\ldots,L\}$. The $W\bs{l}$-component of the sequence in~\eqref{Eqn:DMS} is given to Node~$l$.  The orthogonal \emph{uplink channel} from Node $l$ to the Relay is discrete and memoryless with input alphabet $\set{X}\bs{l}$, output alphabet $\set{Y}\bs{0,l}$, and transition probabilities
\begin{equation*}
q\bs{l}(y|x) \triangleq \mathbb{P}[Y\bs{0,l}=y|X\bs{l}=x].
\end{equation*}
The \emph{downlink broadcast channel} is discrete and memoryless with input alphabet $\set{X}\bs{0}$, output alphabet $\set{Y}\bs{l}$ at Node~$l$, and transition probabilities
\begin{multline*}
q\bs{0}(y_1,y_2,\ldots,y_L|x)
\triangleq \mathbb{P}[Y\bs{1}=y_1,Y\bs{2}=y_2,\\ \ldots, Y\bs{L}=y_L|X\bs{0}=x].
\end{multline*}

A \emph{joint source-channel (JSC) code} with blocklength $n$, see Fig.~\ref{Fig:JSCCArc}, is a collection of $(2L+1)$-maps: the encoder at Node~$l$,
\begin{subequations}\label{Eqn:DefJSCC}
\begin{align}
\phi\bs{l} &:  \bm{\set{W}}\bs{l} \longrightarrow \bm{\set{X}}\bs{l};
\end{align}
the encoder at the Relay,
\begin{align}
\psi\bs{0} &:  \bm{\set{Y}}\bs{0,1} \times \bm{\set{Y}}\bs{0,2} \times \cdots \times \bm{\set{Y}}\bs{0,L}  \longrightarrow \bm{\set{X}}\bs{0};
\end{align}
and the decoder at Node $l$,
\begin{align}
\psi\bs{l} &: \bm{\set{W}}\bs{l} \times \bm{\set{Y}}\bs{l} \longrightarrow \bm{\set{W}}\bs{1} \times \bm{\set{W}}\bs{2} \times \cdots \times \bm{\set{W}}\bs{L}.
\end{align}
\end{subequations}
Node~$l$ observes
\begin{equation*}
\bm{W}\bs{l} \triangleq (W\bs{l}_1,W\bs{l}_2,\ldots,W\bs{l}_n)
\end{equation*}
and transmits
\begin{equation*}
\bm{X}\bs{l} \triangleq \phi\bs{l}(\bm{W}\bs{l}).
\end{equation*}
The Relay observes
\begin{equation*}
\bm{Y}\bs{0,l} \triangleq (Y\bs{0,l}_1,Y\bs{0,l}_2,\ldots,Y\bs{0,l}_n),
\end{equation*}
on the $l$-th uplink channel, and it transmits
\begin{equation*}
\bm{X}\bs{0} \triangleq \psi\bs{0}(\bm{Y}\bs{0,1},\bm{Y}\bs{0,2},\ldots,\bm{Y}\bs{0,L}).
\end{equation*}
Node~$l$ observes
\begin{equation*}
\bm{Y}\bs{l} \triangleq (Y\bs{l}_1,Y\bs{l}_2,\ldots,Y\bs{l}_n)
\end{equation*}
and decodes
\begin{equation*}
(\bm{\hat{W}}\bs{l,1},\bm{\hat{W}}\bs{l,2},\ldots,\bm{\hat{W}}\bs{l,L}) \triangleq \psi\bs{l}(\bm{W}\bs{l},\bm{Y}\bs{l}).
\end{equation*}

\newpage

The \emph{average joint decoding error probability} of a JSC-code is 
\begin{multline}\label{Eqn:JointErrorProb}
P_\text{e} \triangleq \mathbb{P}\Big[(\hat{\bm{W}}\bs{l,1},\ldots,\hat{\bm{W}}\bs{l,L}) \neq (\bm{W}\bs{1},\ldots,\bm{W}\bs{L}) \\
\text{ for one or more }
 l \text{ in } \{1,2,\ldots,L\} \Big].
\end{multline}

\medskip

\begin{definition}
We say that \emph{reliable communication is achievable with JSC codes} if there exists for each $\epsilon > 0$ a code of the form~\eqref{Eqn:DefJSCC} with $P_\text{e} \leq \epsilon$ for some sufficiently large integer $n$.
\end{definition}

\subsection{Main Result}

The following notation is required for the next theorem. Denote the \emph{capacity}~\cite[Eqn.~7.1]{Cover-2006-B} of the $l$-th orthogonal uplink channel by
\begin{equation*}
    \Cu{l} \triangleq \max_{X\bs{l}} I(X\bs{l};Y\bs{0,l}),
\end{equation*}
where the maximisation is over distributions for $X\bs{l}$ on $\set{X}\bs{l}$. If
\begin{equation*}
\set{L}=\big\{l_1,l_2,\ldots,l_{|\set{L}|}\big\}
\end{equation*}
is a nonempty subset of $\{1,2,\ldots,L\}$, then let
\begin{equation*}
W\bs{\set{L}} \triangleq \big(W\bs{l_1},W\bs{l_2},\ldots,W\bs{l_{|\set{L}|}})
\end{equation*}
denote those random variables with indices belonging to $\set{L}$. The next theorem is proved in Appendix~\ref{App:ProofOfThm:JSCC}.

\medskip

\begin{theorem}\label{Thm:JSCC}
Reliable communication is achievable with JSC codes if
\begin{subequations}\label{Eqn:JSCC}
\begin{equation}
\label{Eqn:JSCCUp}
H\big(W\bs{\set{L}}|W\bs{\set{L}^c}\big) < \sum_{l \in \set{L}} \Cu{l}
\end{equation}
holds for each nonempty \emph{strict} subset $\set{L}$ of $\{1,2,\ldots,L\}$, and there exists a distribution for $X\bs{0}$ on $\set{X}\bs{0}$ such that
\begin{equation}
\label{Eqn:JSCCDown}
H\big(W\bs{\sim l}|W\bs{l}\big) < I(X\bs{0};Y\bs{l})
\end{equation}
holds for each $l$ in $\{1,2,\ldots,L\}$.
\end{subequations}
Conversely, if reliable communication is achievable then~\eqref{Eqn:JSCCUp} or~\eqref{Eqn:JSCCDown} hold as inequalities --- instead of strict inequalities --- for some $X\bs{0}$.
\end{theorem}

\subsection{Remarks}

\emph{Non-matched symbol rates:} Theorem~\ref{Thm:JSCC} characterises reliable communication for matched source and channel symbol rates; i.e., $n$ source symbols are mapped to $n$ channel symbols. The proof easily extends to the non-matched symbol rate setting where $n$ source symbols map to $m$ channel symbols.

\medskip

\emph{Networks of Channels:} The uplink condition~\eqref{Eqn:JSCCUp} closely resembles Han's generalisation~\cite[Sec.~1]{Han-1980-A} of the Slepian-Wolf/Cover theorem~\cite{Slepian-Jul-1973-A} to networks of noisy orthogonal channels (see also Barros and Servetto~\cite{Barros-Jan-2006-A}). Indeed, all but one of the inequalities appearing in~\cite[p.~1]{Han-1980-A} also appear as uplink constraints in~\eqref{Eqn:JSCCUp} --- the exception being a total sum rate constraint of the form
\begin{equation}\label{Eqn:HanExtraInequality}
H(W\bs{1},W\bs{2},\ldots,W\bs{L}) \leq \sum_{l=1}^L \Cu{l}.
\end{equation}
Although our problem formulation differs from that of~\cite{Han-1980-A,Barros-Jan-2006-A}, the similarity of these results can be understood by comparing the respective achievability proofs. Han~\cite{Han-1980-A} uses a simple separate source-channel coding argument: he combines an optimal Slepian-Wolf code with optimal channel codes for each orthogonal uplink. In Han's setup, reliable communication is possible if~\eqref{Eqn:JSCCUp} and~\eqref{Eqn:HanExtraInequality} both hold. The uplink part of our proof essentially  uses the same argument, except we do not require that~\eqref{Eqn:HanExtraInequality} holds; i.e., we use fewer bins and, as a consequence, the Relay cannot decode the sources. Indeed, in our setup, the Relay needs only to recover the Slepian-Wolf bin indices and not the individual source sequences. The downlink achievability proof requires JSC coding and is discussed next. 

\medskip

\emph{Joint Source-Channel Coding:} The (downlink) achievability proof of Theorem~\ref{Thm:JSCC} is based on a JSC random-coding argument that builds upon the virtual binning idea developed by Tuncel in~\cite{Tuncel-Apr-2006-A}. To illustrate why the virtual binning approach is useful, momentarily suppose that the Relay is given the entire source $L$-tuple $(\bm{W}\bs{1},\bm{W}\bs{2},\ldots,\bm{W}\bs{L})$ and consider the downlink phase in isolation. With the setup of~\cite{Tuncel-Apr-2006-A} in mind, we can view $(\bm{W}\bs{1},\bm{W}\bs{2},\ldots,\bm{W}\bs{L})$ as a common message that needs to be reliably decoded by every node. Applying~\cite[Thm.~6]{Tuncel-Apr-2006-A} we immediately see that the common message can be reliably decoded by every node whenever~\eqref{Eqn:JSCCDown} holds.

The basic idea behind the proof of~\cite[Thm.~6]{Tuncel-Apr-2006-A} is to randomly generate a downlink channel codeword ($n$ i.i.d. symbols $\sim P_{X\bs{0}}$) for each and every jointly typical source tuple $(\bm{W}\bs{1},\bm{W}\bs{2},\ldots,\bm{W}\bs{L})$. Upon observing a typical source tuple\footnote{An error is declared if the source is not jointly typical.}, the Relay transmits the corresponding channel codeword. Node~$l$, upon observing the channel output $\bm{Y}\bs{l}$, compiles a list of all those channel codewords that are jointly typical with $\bm{Y}\bs{l}$. The codeword list corresponds to list of typical source sequences, with the same number of elements. We may think of the source list as a \emph{(virtual) random bin} in the sense of the classic Slepian-Wolf Theorem~\cite{Kramer-2008-A}. Node~$l$ looks within this list for a unique source tuple that is jointly typical with its source $\bm{W}\bs{l}$; this search will be successful with high probability whenever~\eqref{Eqn:JSCCDown} holds.

The above argument assumes that the entire source tuple is made available to the Relay, which is not the case  in the multi-way relay network. The key difficulty in proving Theorem~\ref{Thm:JSCC} is to overcome the fact that the Relay only has partial knowledge of the source tuple.

\medskip

\emph{Processing Broadcast Satellite:} The source coding work of Wyner \emph{et al.}~\cite[Thm.~1]{Wyner-Jun-2002-A} is a special case of Theorem~\ref{Thm:JSCC}.

%%%%
%%%%
%%%%
%%%%

\section{Separate Source-Channel Coding}\label{Sec:SSCC}

We now compare the general JSC coding architecture of Section~\ref{Sec:JSCC} to a separate source-channel coding architecture.

\subsection{Channel Coding}\label{Sec:ChannelCoding}

The channel-coding problem of interest is analogous to the JSC-coding problem in Fig.~\ref{Fig:JSCCArc} with one exception: the discrete memoryless source $(\bm{W}\bs{1},\bm{W}\bs{2},\ldots,\bm{W}\bs{L})$ is replaced by $L$-independent random variables $(M\bs{1},M\bs{2},\ldots,$ $M\bs{L})$, where each $M\bs{l}$ is uniformly distributed on $\set{M}\bs{l}$.

A channel code with blocklength $n$ is a collection of maps: the encoder at Node~$l$,
\begin{subequations}\label{Eqn:ChannelCode}
\begin{equation}
\phi\bs{l}_\text{c} : \set{M}\bs{l} \longrightarrow \bm{\set{X}}\bs{l};
\end{equation}
the encoder at the Relay,
\begin{equation}
\psi\bs{0}_\text{c} : \bm{\set{Y}}\bs{0,1} \times \bm{\set{Y}}\bs{0,2} \times \cdots \times \bm{\set{Y}}\bs{0,L} \longrightarrow \bm{\set{X}}\bs{0};
\end{equation}
and the decoder at Node~$l$,
\begin{multline}
\psi\bs{l}_{\text{c}} : \set{M}\bs{l} \times \bm{\set{Y}}\bs{l} \longrightarrow \set{M}\bs{1} \times \set{M}\bs{2} \times \cdots \times \set{M}\bs{L}.
\end{multline}
\end{subequations}
The channel code operates in a manner analogous to the JSC-code: Node~$l$ sends
\begin{equation*}
\bm{X}\bs{l} \triangleq \phi\bs{l}(M\bs{l})
\end{equation*}
and decodes
\begin{equation*}
(\hat{M}\bs{l,1},\hat{M}\bs{l,2},\ldots,\hat{M}\bs{l,L}) \triangleq \psi\bs{l}(M\bs{l},\bm{Y}\bs{l}).
\end{equation*}
The \emph{average joint error probability} of a channel code is defined by
\begin{equation*}
\begin{split}
P_\text{e} \triangleq \mathbb{P}\big[(\hat{M}\bs{l,1},\hat{M}\bs{l,2},&\ldots,\hat{M}\bs{l,L}) \neq (M\bs{1},M\bs{2},\ldots,M\bs{L})\\
&\ \ \text{ for one or more  $l$  in $\{1,2,\ldots,L\}$}\big].
\end{split}
\end{equation*}
The \emph{rate} at which Node~$l$ transmits the message $M\bs{l}$ is defined by
\begin{equation*}
\eta\bs{l} \triangleq \frac{1}{n} \log_2 |\set{M}\bs{l}|.
\end{equation*}

\medskip

\begin{definition}
A nonnegative rate tuple $(r\bs{1},r\bs{2},\ldots,r\bs{L})$ is said to be \emph{achievable} if the following holds: for each $\epsilon > 0$ there exists a channel code of the form~\eqref{Eqn:ChannelCode} with $P_\text{e} \leq \epsilon$ and $\eta\bs{l} \geq r\bs{l} - \epsilon$ for some sufficiently large integer $n$.
\end{definition}

\medskip

\begin{definition}
The \emph{capacity region} $\set{C}$ is the set of all achievable rates.
\end{definition}

\medskip

We now give a single-letter expression for $\set{C}$. Let $\set{C}^*$ denote those nonnegative rate tuples $(r\bs{1},r\bs{2},\ldots,r\bs{L})$ for which
\begin{enumerate}
\item[(i)] the uplink channel capacities satisfy
\begin{equation*}
r\bs{l} \leq C_{\text{up}}\bs{l}
\end{equation*}
for all $l = 1,2,\ldots,L$; and
\item[(ii)] there is a distribution for $X\bs{0}$ on $\set{X}\bs{0}$ such that 
\begin{equation*}
\sum_{l' \neq l} r\bs{l'} \leq I(X\bs{0};Y\bs{l})
\end{equation*}
holds for all $l = 1,2,\ldots,L$.
\end{enumerate}

\medskip

\begin{lemma}\label{Lem:ChannelCoding}
$\set{C} = \set{C}^*$.
\end{lemma}

\medskip

\begin{proof}
The lemma can be proved in the same way as Theorem~\ref{Thm:JSCC} with $(\bm{W}\bs{1},\bm{W}\bs{2},$ $\ldots,\bm{W}\bs{L})$ replaced by $(M\bs{1},M\bs{2},\ldots,$ $M\bs{L})$ and $H(W\bs{\set{L}}|W\bs{\set{L}^c})$  replaced by $\sum_{l \in \set{L}} \eta\bs{l}$. We omit the technical details.
\end{proof}

\begin{figure}
\begin{center}
\includegraphics[width=0.5\textwidth]{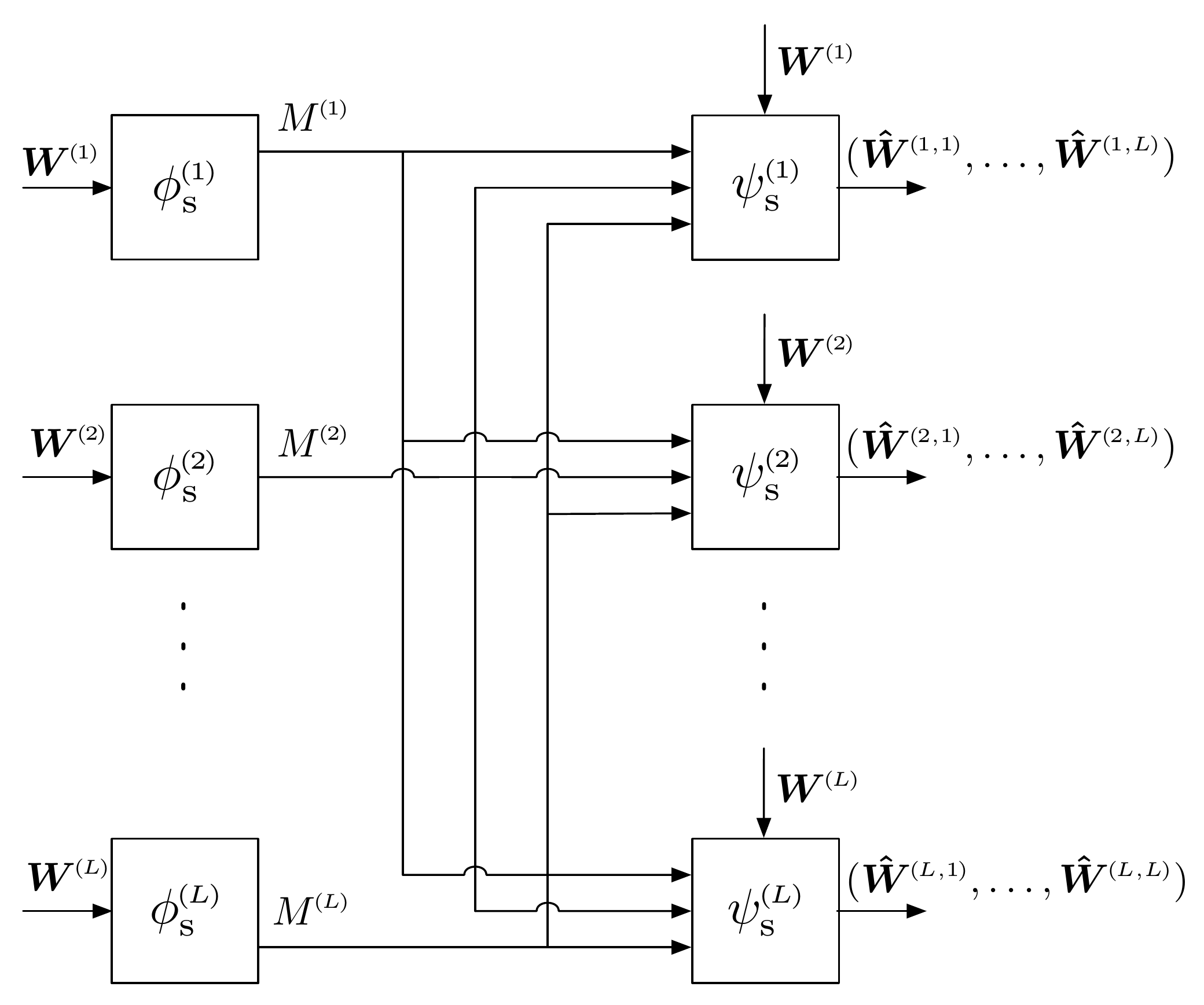}
\caption{The multi-way relay network source coding problem. The source code is designed on the assumption that the $L$-indices $(M\bs{1},M\bs{2},\ldots,M\bs{L})$ can be reliably transported over the network by a channel code.}
\label{Fig:SourceCodingArc}
\end{center}
\end{figure}

\subsection{Source Coding}\label{Sec:SourceCoding}

The source-coding problem of interest is the multi-source multicast problem shown in Fig.~\ref{Fig:SourceCodingArc}. The problem is the source coding counterpart of the channel coding problem of Section~\ref{Sec:ChannelCoding} in the following sense: a source code from this section combined with a channel code from the Section~\ref{Sec:ChannelCoding} produces a (separate source-channel) code for the overall JSC problem.

\begin{subequations}\label{Eqn:SourceCodeDefinition}
A source code of length $n$ is a collection of $2L$-maps: the compressor at Node~$l$,
\begin{equation}
\phi_\text{s} : \bm{\set{W}}\bs{l} \longrightarrow \set{M}\bs{l};
\end{equation}
and the decompressor at Node~$l$,
\begin{multline}
\psi_\text{s} : \bm{\set{W}}\bs{l} \times \set{M}\bs{1} \times \set{M}\bs{2} \times \cdots \times \set{M}\bs{L} \\ \longrightarrow \bm{\set{W}}\bs{1} \times \bm{\set{W}}\bs{2} \times \cdots \times \bm{\set{W}}\bs{L}.
\end{multline}
\end{subequations}
Node~$l$ sends 
\begin{equation*}
M\bs{l} \triangleq \phi_\text{s}\bs{l}(\bm{W}\bs{l})
\end{equation*}
and decompresses
\begin{equation*}
(\bm{\hat{W}}\bs{l,1},\bm{\hat{W}}\bs{l,2},\ldots,\bm{\hat{W}}\bs{l,L}) \triangleq \psi\bs{l}_\text{s}(M\bs{1},M\bs{2},\ldots,M\bs{L}).
\end{equation*}
The \emph{joint decoding error probability} is defined in the same way as~\eqref{Eqn:JointErrorProb}, and the \emph{compression rate} of Node~$l$ is defined by
\begin{equation*}
\eta_\text{s}\bs{l} \triangleq \frac{1}{n} \log_2 |\set{M}\bs{l}|.
\end{equation*}

\medskip

\begin{definition}
A nonnegative rate tuple $(r\bs{1},r\bs{2},\ldots,r\bs{L})$ is said to be \emph{achievable} if the following holds: for each $\epsilon > 0$ there exists a source code of the form~\eqref{Eqn:SourceCodeDefinition} with $n$ sufficiently large, $P_\text{e} \leq \epsilon$ and $\eta_\text{s}\bs{l} \leq r\bs{l} + \epsilon$ for all $l = 1,2,\ldots,L$.
\end{definition}

\medskip

\begin{definition}
The \emph{source-coding rate region} $\set{R}$ is defined as the set of all achievable rate tuples.
\end{definition}

\medskip

\begin{lemma}\label{Lem:SourceCoding}
$\set{R}$ is equal to the set of all nonnegative rate tuples $(r\bs{1},r\bs{2},$ $\ldots,r\bs{L})$ for which
\begin{equation*}
\sum_{l \in \set{L}} r\bs{l} \geq H(W\bs{\set{L}}|W\bs{\set{L}^c})
\end{equation*}
holds for each nonempty and strict subset $\set{L}$ of $\{1,2,\ldots,L\}$.
\end{lemma}

\medskip

\begin{proof}
The lemma is a simple consequence of the Slepian-Wolf/Cover theorem~\cite{Slepian-Jul-1973-A,Cover-Mar-1975-A}. The details are omitted.
\end{proof}

\subsection{Separate Source and Channel Coding}\label{Sec:Separation}

Reliable communication with separate source and channel codes is possible if the intersection of the interior of $\set{R}$ and the interior of $\set{C}$ is nonempty.

\medskip

\begin{theorem}\label{Thm:SSCC}
Reliable communication with separate source and channel codes is possible if there exists nonnegative rates $(r\bs{1},r\bs{2},\ldots,r\bs{L})$ such that
\begin{subequations}\label{Eqn:Separation}
\begin{equation}
H(W\bs{\set{L}}|W\bs{\set{L}^c}) < \sum_{k \in \set{L}} r\bs{k} < \sum_{l \in \set{L}} \Cu{k}
\end{equation}
holds for each nonempty and strict subset $\set{L}$ of $\{1,2,\ldots,L\}$ and
\begin{equation}
\sum_{k \neq l} r\bs{k} < I(X\bs{0};Y\bs{l})
\end{equation}
holds for all  $l$ in $\{1,2,\ldots,L\}$.
\end{subequations}
\end{theorem}

\medskip

\begin{proof}
The theorem is an immediate consequence of Lemmas~\ref{Lem:ChannelCoding} and~\ref{Lem:SourceCoding}.
\end{proof}

\medskip

Separate source and channel coding is optimal for two nodes. Specifically, the achievability assertion of Theorem~\ref{Thm:SSCC} is equivalent to that of Theorem~\ref{Thm:JSCC}: if~\eqref{Eqn:JSCC} holds, then we can find $(r\bs{1},r\bs{2})$ simultaneously satisfying
\begin{align*}
H(W\bs{1}|W\bs{2}) &< r\bs{1} < \Cu{1}\\
H(W\bs{2}|W\bs{1}) &< r\bs{2} < \Cu{2}
\end{align*}
and
\begin{align*}
r\bs{1} &< I(X\bs{0};Y\bs{2})\\
r\bs{2} &< I(X\bs{0};Y\bs{1}).
\end{align*}

The situation is more complicated for three or more nodes. Indeed, the achievability assertion of Theorem~\ref{Thm:SSCC} is more restrictive than that of Theorem~\ref{Thm:JSCC} in general; in particular, it is not possible to simultaneously remove all redundancies in the sources for every node, and such redundancies can be exploited by the channel code. We describe such a situation in Appendix~\ref{App:Example}.

It should be noted that the capacity region $\set{C}$ is formulated with the requirement that each node reliably communicates a single message to every other node. A more general setup would permit the use of a \emph{private} message from each node to each subset of nodes. Characterising the resultant $L(2^{L-1}-1)$-dimensional capacity region appears to be a formidable task; in particular, the problem includes the setup of~\cite{Oechtering-Jul-2012-C} as a special case. Theorem~\ref{Thm:SSCC} should therefore be understood as a sufficient condition for separate source and channel coding to be optimal. Finally, it is interesting to juxtapose such difficulties to the relatively simple JSC coding scheme used to prove Theorem~\ref{Thm:JSCC}.

%%%%
%%%% Section 4: Practical Codes
%%%%

\section{Practical Codes}\label{Sec:Practical} 

We now consider the problem of designing practical, low complexity, codes that can approach those theoretical limits established for JSC coding in Sections~\ref{Sec:JSCC} and~\ref{Sec:SSCC}. Iterative error correction codes have been extensively investigated for JSC coding; for example, see \cite{Garcia-Frias-Sep-2001-A} on joint turbo decoding and estimation of hidden Markov sources, or \cite{Xu-May-2007-A} on distributed JSC coding of video. In this section, we present a JSC coding scheme for the two-way relay network that is based on low-density parity-check (LDPC) codes~\cite{Gallager-Jan-1962-A,MacKay-Mar-1999-A,Johnson-2010-B,Richardson-2008-B}. In particular, we consider multi-edge LDPC codes~\cite{Richardson-2008-B} that are widely used in applications such as wiretap and multi-relay channels \cite{Rathi-Nov-2009-C,Azmi-Nov-2011-A}. Our aim is to show how the individual codes for source and channel coding can be represented by a joint factor-graph \cite{Kschischang-Feb-2001-A}, and we use this graph to provide an alternative view of the separation principle using multi-edge density evolution.

In the following, we represent the channel coding message of Node $l$ (denoted by $M\bs{l}$ in Section~\ref{Sec:ChannelCoding}) using the binary notation $\mbf{B}\bs{l}$. As there are only two nodes, we will denote the node that is not Node $l$ as Node $\sim l$. 

The encoder at Node $l$ uses a linear source code to compress its source vector $\mbf{W}\bs{l}$ as
\begin{align*}
    \mbf{B}\bs{l} = \mbf{W}\bs{l}{\pchk_\text{s}\bs{l}}^{\mathrm{T}}.
\end{align*}
Here $\mbf{W}\bs{l}$ and $\mbf{B}\bs{l}$ are binary vectors of length $n$ and $r\bs{l}n$ respectively, and $\pchk_\text{s}\bs{l}$ is the parity-check matrix defining the source code of Node $l$.

The compressed vector $\mbf{B}\bs{l}$ is mapped to a channel codeword $\mbf{X}\bs{l}$, which is transmitted over the uplink channel. The uplink channel code is defined by a parity-check matrix $\pchk_\text{c}\bs{l}$\footnote{For more information on how to encode a message $\mbf{B}\bs{l}$ into a codeword $\mbf{X}\bs{l}$ for a particular code $\pchk_\text{c}\bs{l}$ see \cite[App. A]{Richardson-2008-B}.}. The relay decodes $\mbf{B}\bs{l}$ from the noisy outputs $\mbf{Y}\bs{0,l}$ of the $l$-th uplink channel. It then maps the concatenation of $\mbf{B}\bs{1}$ and $\mbf{B}\bs{2}$ to $\mbf{X}\bs{0}$  using a channel encoder defined by the parity-check matrix $\pchk_\text{c}\bs{0}$. Each node uses a channel decoder to recover the other user's index $\mbf{B}\bs{\sim l}$ using their own index $\mbf{B}\bs{l}$ and the noisy observation $\mbf{Y}\bs{l}$ from the broadcast channel. Finally, the source code is decoded separately resulting in an estimate of $\mbf{W}\bs{\sim l}$. This system is represented by the factor-graph in Fig.~\ref{fig:factorgraph}.

\begin{figure*}[t]
    \centering
    \includegraphics[width=1.1\columnwidth]{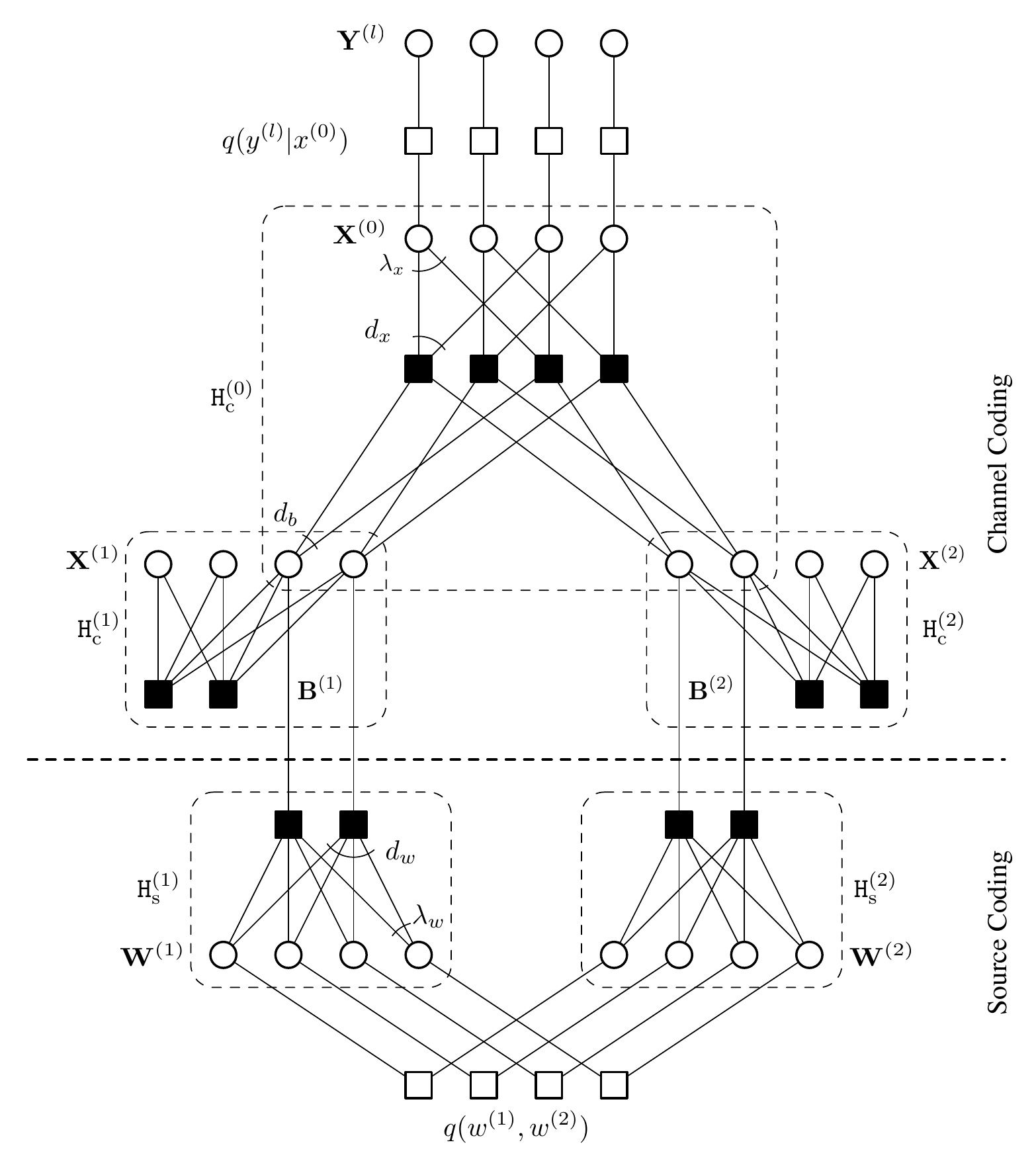}
    \caption{Factor-graph of the overall system. Circles represent (binary) symbols; parity-check equations are represented by solid black squares. The correlation between the symbols and the channel transition probability for the downlink are represented by empty squares.}
    \label{fig:factorgraph}
\end{figure*}

\subsection{Source-Channel Separation from a Factor-Graph Viewpoint} 
In Section~\ref{Sec:Separation}, we showed that separate source-channel coding is asymptotically (in blocklength) optimal in the two-way relay network. We now provide an alternative view of this result using the factor-graph representation in Fig.~\ref{fig:factorgraph}. Let us first assume that the relay successfully decodes each bin index -- the uplink code does not add to the discussion on separate versus joint decoding of the source and downlink codes. Consider the graph shown in Fig.~\ref{fig:factorgraph} depicting the source codes and downlink channel code; these codes are connected via the bin indices.

A separate source-channel decoder will apply the channel decoder to determine the other user's index $\mbf{B}\bs{\sim l}$ from $\mbf{Y}\bs{l}$ and then separately apply the source decoder to estimate $\mbf{W}\bs{\sim l}$ from $\mbf{B}\bs{\sim l}$. A joint source-channel decoder will decode $\mbf{W}\bs{\sim l}$ directly from $\mbf{Y}\bs{l}$ by decoding the source and channel codes on the joint factor graph, and exchanging soft (extrinsic) information between the two parts of the factor graph.

Consider now the case when the source is compressed with a rate that equals the conditional entropy. In the channel coding setting this corresponds to a capacity achieving code. It has been shown in \cite{Peleg-Nov-2006-C} that the extrinsic information about the coded bits of any good (capacity achieving) code is zero above capacity. This implies that a joint decoder cannot outperform a separate decoding scheme in such a setting since there is no extrinsic information available to the joint decoder. Therefore, source-channel separation with separate decoding can be seen to be optimal from a factor graph perspective if capacity achieving source codes are applied.

In the following we will investigate source-channel separation in a practical setting by designing source and channel codes for the two-way relay channel and considering their performance using both density evolution and finite length simulations.

\subsection{Design of the Source and Channel Codes}

The task of the source code is to map the message $\mbf{W}^{(l)}$ of length $n$ bits to a message $\mbf{B}^{(l)}$ of length $r\bs{l}n$ bits such that the other node can reconstruct the message using its own message as side information. This is a Slepian-Wolf coding problem \cite{Slepian-Jul-1973-A} with bin index $\mbf{B}^{(l)}$. An optimal Slepian-Wolf code can be realised by using the \emph{syndrome} of a linear code, which is optimised for a particular symmetric dual channel~\cite{Wyner-Jan-1974-A,Chen-Sep-2009-A}. The optimisation of the degree distribution of an LDPC code for a symmetric channel with uniform input is well studied~\cite{Richardson-Feb-2001-A,Richardson-2008-B} and we will design our source codes in this way.

For the downlink the relay has to communicate the messages $\mbf{B}\bs{1}$ and $\mbf{B}\bs{2}$ to both nodes simultaneously by broadcasting a codeword $\mbf{X}\bs{0}$. Our proposed coding scheme is as follows. First, the relay treats the concatenation of $\mbf{B}\bs{1}$ and $\mbf{B}\bs{2}$ as the systematic part of an LDPC code. It then uses an LDPC encoder to determine the message $\mbf{X}\bs{0}$ (the ``parity'' bits), which it transmits over the broadcast channel. At each node, the channel decoder knows its own message $\mbf{B}\bs{l}$ and treats the message of the other node as being erased.

The overall structure of the parity-check matrix $\pchk_\text{c}\bs{0}$ of our scheme is shown in Fig.~\ref{fig:bcmatrix}. This matrix consists of a concatenation of the matrices $\pchk\bs{1}$, $\pchk\bs{2}$ (corresponding to $\mbf{B}\bs{1}$, $\mbf{B}\bs{2}$), and the matrix $\pchk\bs{0}$ (corresponding to $\mbf{X}\bs{0}$).
\begin{figure}[t]
    \centering
    \includegraphics[width=1\columnwidth]{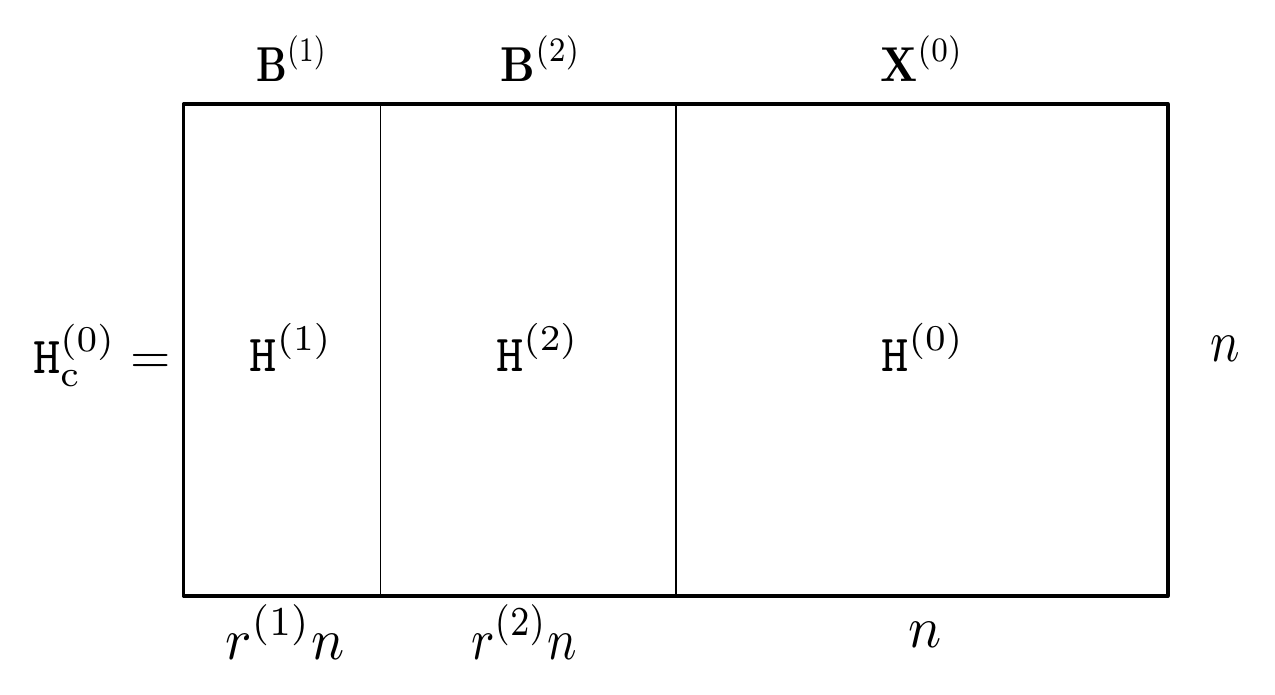}
    \caption{Structure of the parity-check matrix for the downlink channel code.}
    \label{fig:bcmatrix}
\end{figure}
The codebook used by the relay is therefore defined as
\begin{multline*}
    \set{C}\bs{0} = \Big\{[\mbf{b}\bs{1} \mbf{b}\bs{2} \mbf{x}\bs{0}] \in \{0,1\}^{(r\bs{1}n+r\bs{2}n+n)} \\
    : [\mbf{b}\bs{1} \mbf{b}\bs{2} \mbf{x}\bs{0}] {\pchk_\text{c}\bs{0}}^{\mathrm{T}} = \mbf{0} \Big\}.
\end{multline*}
The relay can determine its broadcast message $\mbf{X}\bs{0}$ by solving
\begin{align*}
    \mbf{B}\bs{1}{\pchk\bs{1}}^{\mathrm{T}} \oplus \mbf{B}\bs{2}{\pchk\bs{2}}^{\mathrm{T}} \oplus \mbf{X}\bs{0}{\pchk\bs{0}}^{\mathrm{T}} = \mbf{0},
\end{align*}
where $\oplus$ denotes the element-wise addition over GF(2). An efficient algorithm to determine $\mbf{X}\bs{0}$ is described in \cite[App. A]{Richardson-2008-B}.

To allow successful encoding at the relay and decoding at the nodes:
\begin{itemize}
\item the square matrix $\pchk\bs{0}$ has to be of full rank (over GF(2)) (this enables the relay to determine the vector $\mbf{X}\bs{0}$ given the messages $\mbf{B}\bs{1}$ and $\mbf{B}\bs{2}$),
\item the matrices $\pchk\bs{1}$ and $\pchk\bs{2}$ have to be free of stopping sets (a set of variable nodes that, if erased, cannot be resolved by a message-passing decoder \cite[Sec. 3.22]{Richardson-2008-B}).
\end{itemize}
The first condition is obvious since the relay cannot determine its broadcast message if the square matrix $\pchk\bs{0}$ is rank deficient. To show the necessity of the second condition, assume that the broadcast channel is noiseless, i.e., the nodes know $\mbf{X}\bs{0}$ without errors. The nodes now consider their own message $\mbf{B}\bs{l}$ as being known and solve for the other message $\mbf{B}\bs{\sim l}$ that is treated as being erased. If this set of erased variables contains a stopping set, then the iterative decoder is unable to solve for all erasures and gets stuck in the largest stopping set.\footnote{We note that this constraint only applies to the (sub-optimal) message passing decoder which we consider in this paper. For a maximum likelihood decoder we would have to require that the erased variable nodes do not contain the support of a codeword. The support of any codeword is a stopping set but not all stopping sets correspond to codewords. Therefore, the constraint for the message passing decoder is more restrictive than that for the maximum likelihood decoder.}

The absence of stopping sets can also be expressed by requiring that the sub-matrices $\pchk\bs{1}$ and $\pchk\bs{2}$ can be converted to triangular form using only row and column swaps. This corresponds to the encoding problem of LDPC codes as described in detail in \cite[App. A]{Richardson-2008-B} which allows us to apply their results. In particular, we use \cite[Thm. A.16]{Richardson-2008-B} to design LDPC codes which can be converted to triangular form.

An optimised code for the downlink phase has to be designed for both nodes simultaneously, i.e., it has to perform close to the respective theoretical limit for both downlink channels. Such a code can be analysed and optimised using multi-edge type density evolution \cite[Sec. 7]{Richardson-2008-B}.

\subsection{Density Evolution and Numerical Examples}
\label{sec:de}
We apply multi-edge type density evolution to analyse the performance of the coding structure introduced in the previous section. We make the following assumptions for our numerical examples:
\begin{itemize}
\item $(W\bs{1},W\bs{2})$ is a doubly symmetric binary source with cross-over probability $0 < \rho < 1/2$, i.e., 
\begin{equation*}
q_\text{s}(w\bs{1},w\bs{2}) = 
\left\{ 
\begin{array}{ll}
(1/2) \cdot(1-\rho) & \text{ if } w\bs{1} = w\bs{2}\\
(1/2) \cdot \rho &\text{ otherwise.} 
 \end{array}
\right.
\end{equation*}
The conditional entropy is therefore 
\begin{equation*}
H(W^{(1)}|W^{(2)})=H(W^{(2)}|W^{(1)})=h(\rho),
\end{equation*}
where $h(\cdot)$ denotes the binary entropy function.
\item The (orthogonal) uplink channels are binary input additive white Gaussian noise (BIAWGN) channels with noise variance $\sigma_\text{u}^{2}$ each.
\item The downlink channels are modelled as BIAWGN channels with noise variance $\sigma_\text{d}^{2}$ each.
\end{itemize}

Furthermore assume that the noise variance $\sigma_\text{u}^{2}$ of the uplink channel is fixed and the nodes can communicate their messages reliably to the relay. For the downlink we are interested in the pairs $(\rho, \sigma_\text{d}^{2})$ which define the achievable region
\begin{align}
\label{equ:region}
\set{D} \triangleq \left\{ (\rho,\sigma_\text{d}^{2})\in (0,1/2) \times \mathbb{R}_+ : \lim_{\ell \to \infty} P_{\rm DE}(\ell) < \epsilon \right\},
\end{align}
where $\epsilon > 0$ is an arbitrarily small constant and $P_{\rm DE}(\ell)$ denotes the bit error probability of the estimate $\mbf{\hat{W}}\bs{\sim l}$ after $\ell$ iterations of message passing decoding, i.e., $\set{D}$ consists of all pairs of source correlation $\rho$ and noise variance $\sigma_\text{d}^{2}$ on the downlink channel where the iterative decoder converges to an arbitrarily small error probability for sufficiently large block length.

Since we assume that the uplink channels can be decoded by the relay, we focus on the decoding problem at the nodes where we have the constraint:
\label{equ:con}
\begin{align}
h(\rho) \leq r\bs{l} \leq C_{\text{BIAWGN}}(\sigma_\text{d}^{2})\label{equ:conjoint}
\end{align}
where $C_{\text{BIAWGN}}(\sigma^{2})$ denotes the capacity of a BIAWGN channel with noise variance $\sigma^{2}$.

\begin{table*}
    \begin{center}
    \caption{Degrees and degree distributions of all codes used for the numerical examples. See Fig.~\ref{fig:factorgraph} for the definition of the degrees. \label{Table:LDPC_degrees}}
    \begin{tabular}{c c c c c c c c c c}
        \toprule
        \multicolumn{4}{c}{Source codes} & \multicolumn{6}{c}{Channel codes downlink}\\
        \multicolumn{4}{c}{} & \multicolumn{2}{c}{separate} & \multicolumn{2}{c}{separate} & \multicolumn{2}{c}{joint}\\
        \multicolumn{2}{c}{rate $1/2$} & \multicolumn{2}{c}{rate $1/4$} &
        \multicolumn{2}{c}{rate $1/2$} & \multicolumn{2}{c}{rate $1/4$} & \multicolumn{2}{c}{rate $1/2$}\\
        degree $i$  & $\lambda_{w,i}$ &     degree $i$  & $\lambda_{w,i}$ &
        degree $i$  & $\lambda_{x,i}$ &     degree $i$  & $\lambda_{x,i}$ &     degree $i$  & $\lambda_{x,i}$\\
        \cmidrule(lr){1-2}  \cmidrule(lr){3-4} \cmidrule(lr){5-6} \cmidrule(lr){7-8} \cmidrule(lr){9-10}
        $2$     & $0.1710$  & $2$     & $0.1046$  & $2$     & $0.3657$  & $2$     & $0.3503$  & $2$     & $0.5254$\\
        $3$     & $0.2075$  & $3$     & $0.1984$  & $3$     & $0.1203$  & $3$     & $0.0731$  & $3$     & $0.1612$\\
        $8$     & $0.0800$  & $5$     & $0.1189$  & $12$    & $0.0963$  & $5$     & $0.0161$  & $18$    & $0.1349$\\
        $9$     & $0.2657$  & $6$     & $0.0006$  & $13$    & $0.1797$  & $6$     & $0.2043$  & $19$    & $0.1785$\\
        $47$    & $0.1864$  & $9$     & $0.1597$  & $45$    & $0.0162$  & $19$    & $0.0761$  & $$      & $$\\
        $48$    & $0.0894$  & $10$    & $0.0616$  & $46$    & $0.2218$  & $20$    & $0.0370$  & $$      & $$\\
        $$      & $$        & $19$    & $0.0458$  & $$      & $$        & $33$    & $0.1991$  & $$      & $$\\
        $$      & $$        & $20$    & $0.0453$  & $$      & $$        & $34$    & $0.0440$  & $$      & $$\\
        $$      & $$        & $24$    & $0.1881$  & $$      & $$        & $$      & $$        & $$      & $$\\
        $$      & $$        & $25$    & $0.0770$  & $$      & $$        & $$      & $$        & $$      & $$\\
        \cmidrule(lr){1-2}  \cmidrule(lr){3-4} \cmidrule(lr){5-6} \cmidrule(lr){7-8} \cmidrule(lr){9-10}
        $d_{w}$ & $10$      & $d_{w}$ & $22$      & $d_{b}$ & $3$       & $d_{b}$ & $3$       & $d_{b}$ & $3$\\
        $$      & $$        & $$      & $$        & $d_{x}$ & $4$       & $d_{x}$ & $4$       & $d_{x}$ & $3$\\
        \bottomrule
        \end{tabular}
    \end{center}
\end{table*}
\medskip

\subsubsection{Example 1 (Separate Decoding)}

The decoder at Node $l$ first decodes the other node's bin index $\mbf{B}\bs{\sim l}$. After performing a hard decision on the bin index, the node decodes the Slepian-Wolf code $\pchk_{s}\bs{\sim l}$ using the other node's bin index $\mbf{B}\bs{\sim l}$ and its own message $\mbf{W}\bs{l}$ as side information.

In this example we optimised source and channel codes for two cases 
\begin{equation*}
r\bs{1}=r\bs{2}=1/2\quad  \text{and}\quad r\bs{1}=r\bs{2}=1/4.
\end{equation*}
The resulting source and channel codes are given in Table~\ref{Table:LDPC_degrees}. In Fig.~\ref{fig:regionsep} we have marked the optimised points for both rates (by a cross and a square respectively). We can see that the (density evolution) performance of these optimised codes is close to the theoretical limits. Note that due to practical constraints, for example we considered only LDPC codes with a maximum node degree of $50$, these codes are capacity approaching rather than capacity achieving and so their performance is still bounded away from capacity, as shown by the small gap to the (1/2,1/2) and (1/4,1/4) points respectively.

\medskip

\subsubsection{Example 2 (Decoding when the Source and Channel Vary)}

In this example we consider the codes from Example 1 when the capacity of the downlink channel is
higher and/or when the conditional entropy of the source is lower than the optimised values. In the case of separate source and channel coding this gives the rectangular regions shown in Fig.~\ref{fig:regionsep}.

We next consider the same source and channel codes we used in the rate-1/2 case in Example 1, but now apply a joint decoder. This is achieved by applying the sum-product algorithm \cite{Kschischang-Feb-2001-A} to the joint factor graph consisting of the downlink code and the Slepian-Wolf source code, i.e., the graph shown in Fig.~\ref{fig:factorgraph} assuming error-free uplink channels. The achievable regions for this joint decoder are shown in Fig.~\ref{fig:regionjoint} and compared to the separate decoder (solid and dashed lines, respectively).

When the conditional entropy of the source is less than the value for which the system was designed, the source code is of course no longer optimal. A separate decoder ignores the source when decoding the downlink channel code and therefore it cannot exploit any remaining redundancy after decoding the source. However, a decoder that \emph{jointly decodes} the downlink channel code and the source code can exploit this remaining redundancy. The achievable region of the joint decoder is thus larger than that of the separate decoder due to its improved performance when the conditional entropy of the source is less than 1/2 at the same time that the capacity of the downlink channel is less than 1/2.

The joint source-channel decoder can therefore be seen to be more robust to variations of the source and channel away from the design entropy and rate. So while separate decoding is optimal for the particular design rate/entropy (here $1/2$,$1/2$), joint decoding is more robust when the source and channel vary. Indeed, a close observation of Fig.~\ref{fig:regionjoint} reveals that in this particular example the joint decoding scheme can decode successfully at lower downlink capacities than the separate scheme even for the ($1/2$,$1/2$) point where the separate scheme is optimised. This is because the designed source code is not capacity achieving (see the comment in Example 1) so it is not compressing at the theoretical limit and there is still some remaining redundancy to be exploited by the joint decoder.

\medskip

\subsubsection{Example 3 (Code Optimisation for Joint Decoding)}

In the previous example we presented the achievable region of a joint decoder but we used the same codes as in Example 1, i.e., codes that have been optimised for separate decoding of a specific source. Now we use the same source code as above but instead of optimising the channel code for the downlink for one particular source we optimise it for a range of sources. In particular, we chose to minimise the area between the achievable region and the theoretical limit over the entire range of source entropies. The results of such an optimisation process are shown in Fig.~\ref{fig:regionjoint} (dash-dotted line). We observe that such an optimised scheme performs close to the limit over a wide range. However, this comes with a loss (at the rate 1/2 point) compared to a coding scheme which is optimised for that particular setting.

\medskip

\subsubsection{Example 4 (Finite Length Results)}

In addition to the asymptotic results of the previous examples we present finite length results. For this purpose we consider the codes of Examples 1 and 2, i.e., codes of rate $1/2$ which have been optimised for a separate decoder. Finite length codes have been constructed for source blocks of length $n=10^4$ using the progressive edge growth (PEG) algorithm \cite{Xiao-Yu-Jun-2002-C}. For the simulations shown in Fig.~\ref{fig:wer} we fixed the cross-over probability $\rho$ of the source and varied the signal-to-noise ratio ($E_s/N_0$) of the downlink channel (on the horizontal axis). We repeat this process for three separate sources ($\rho = 0.09, 0.07$ and $0.05$). Results are shown for a separate and joint decoder (solid and dashed lines, respectively).

First, consider the case where $\rho=0.09$. This cross-over probability is close to the threshold of the source code ($\rho_{\text{th}} = 0.1064$) which leads to a high probability of error of the finite length source code (at approximately $10^{-2}$). These errors are independent of whether a separate or joint decoder is used and are because the finite length source code is far from capacity achieving. However, since the source cross-over probability is still slightly below the threshold, a joint decoder can exploit the remaining redundancy and can decode at a slightly lower $E_s/N_0$ for the downlink channel at bit error rates above this error floor.

\begin{figure}[t]
    \centering
    \includegraphics[width=1\columnwidth]{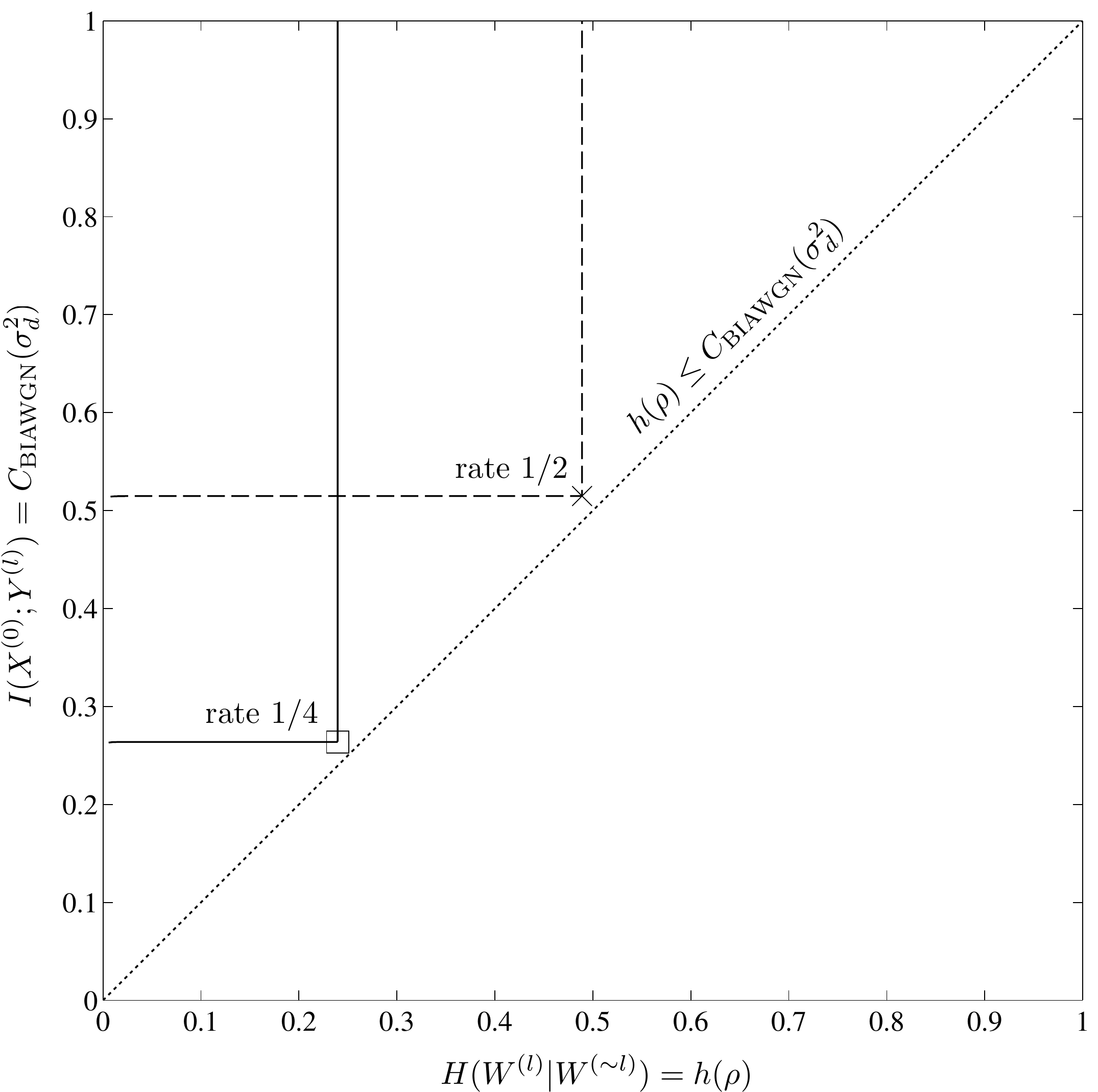}
    \caption{Achievable pairs of $\rho$ and $\sigma_{d}^{2}$ in terms of $h(\rho)$ and $C_{\text{BIAWGN}}(\sigma_{d}^{2})$ for rates $1/2$ (dashed line) and rates $1/4$ (solid line). The dotted line represents the theoretical limit \eqref{equ:conjoint}. The uplink from the nodes to the relay is assumed to be error-free.}
    \label{fig:regionsep}
\end{figure}

\begin{figure}[t]
    \centering
    \includegraphics[width=1\columnwidth]{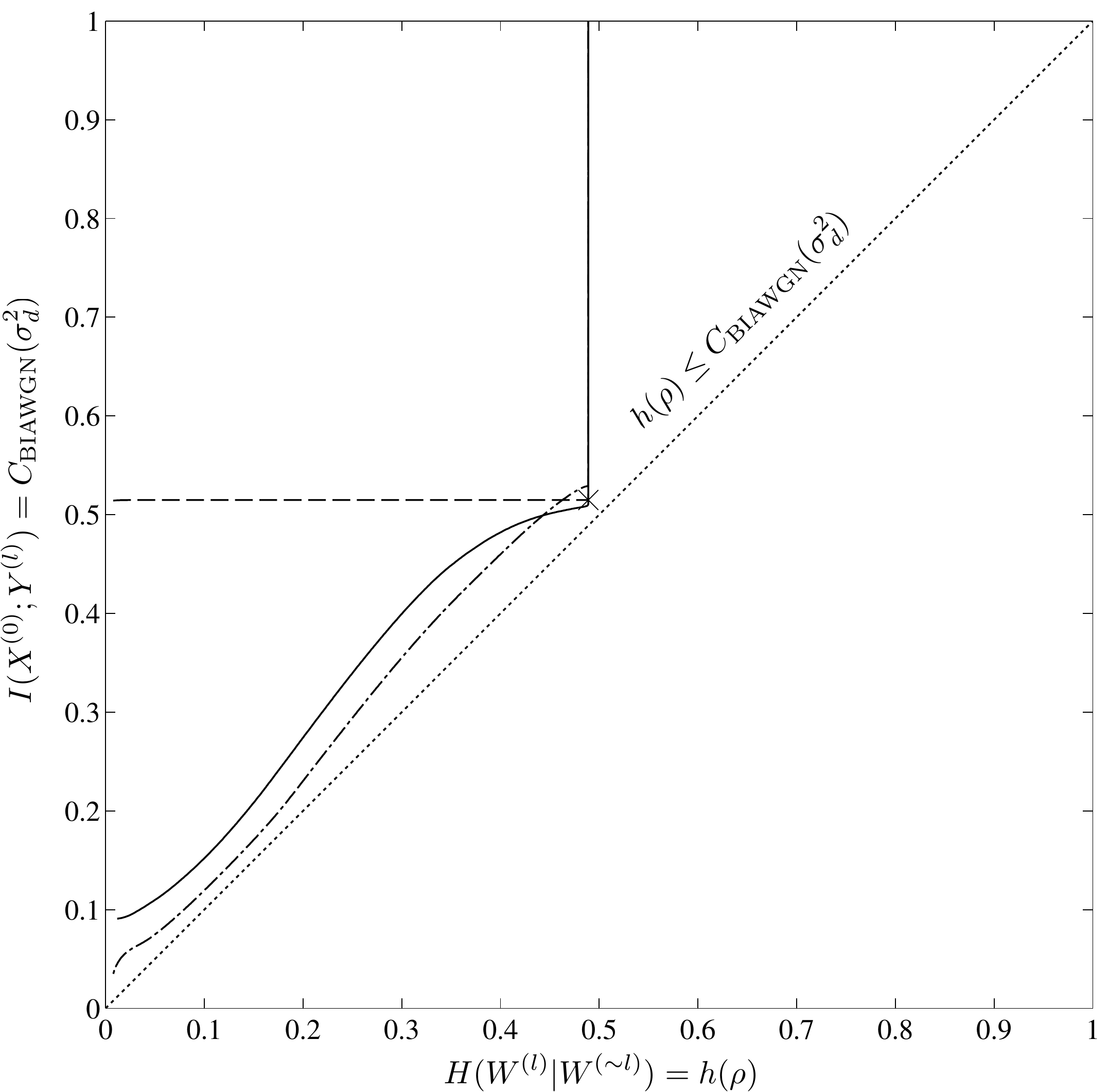}
    \caption{Achievable pairs of $\rho$ and $\sigma_{d}^{2}$ in terms of $h(\rho)$ and $C_{\text{BIAWGN}}(\sigma_{d}^{2})$ for rates $1/2$. The dashed line corresponds to a separate decoder (Example 1) and the solid line corresponds to a joint decoder operating on the code designed for a separate scheme (Example 2). The dash-dotted line represents a joint decoder where the channel code is optimised to perform well over a wide range of sources (Example 3).}
    \label{fig:regionjoint}
\end{figure}

The error floor caused by the source code is decreased when the sources have a smaller cross-over probability (i.e. the sources have a greater correlation). For example, lowering the cross-over probability to $\rho=0.07$ or $\rho=0.05$ shown in Fig.~\ref{fig:wer} leads to an error probability of the source code which is below our simulation range. In these cases a joint decoder will exploit the remaining redundancy and is able to decode at a significantly lower $E_s/N_0$ for the downlink channel than the separate decoder. The channel coding part of the separate decoder cannot benefit from a lower $\rho$ and so has the same downlink channel performance in all three cases.

This leads to the conclusion that, while source/channel separation is optimal in an asymptotic setting of infinite block length, any practical finite length system will benefit from a joint decoder.

\begin{figure*}[t]
    \centering
    \includegraphics[width=1.9\columnwidth]{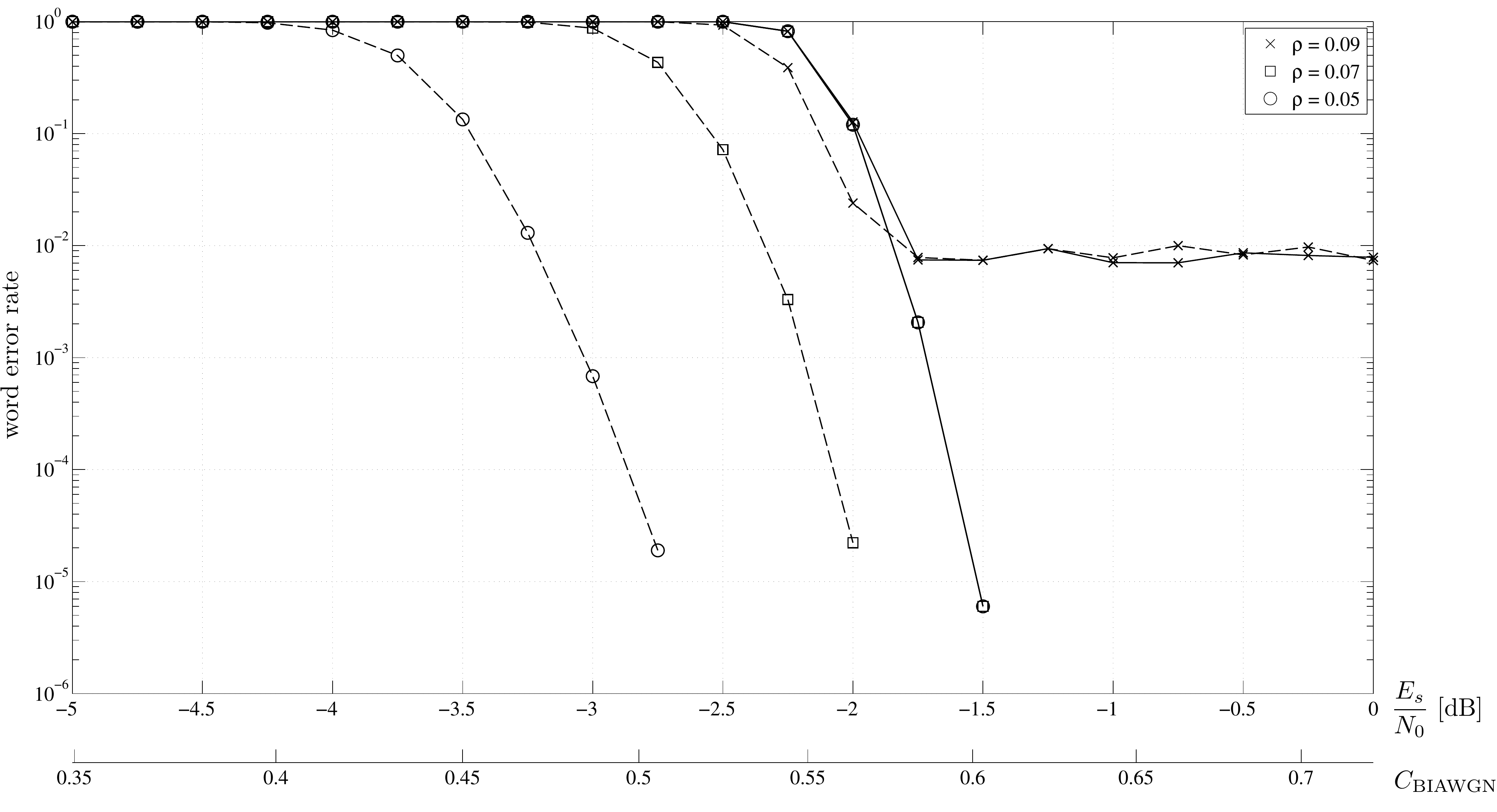}
    \caption{Word error rates of separate (solid lines) and joint decoding (dashed lines) as a function of the signal-to-noise ratio of the downlink channel. On the horizontal axis we also show the capacity of the downlink channel for easier comparison with the achievable regions.}
    \label{fig:wer}
\end{figure*}

%%%%
%%%% Appendix
%%%%

\appendices

\section{Proof of Theorem~\ref{Thm:JSCC}}\label{App:ProofOfThm:JSCC}

In the usual way, we split the proof of Theorem~\ref{Thm:JSCC} into two parts: the \emph{achievability} assertion and the \emph{converse} assertion.

\subsection{Proof of Theorem~\ref{Thm:JSCC} (Achievability)}

\subsubsection{Proof Outline}

The main elements of the proof are best understood for the special case where the uplink channels are noiseless. Extending the proof to the more general case claimed in Theorem~\ref{Thm:JSCC} is straightforward.

The (noiseless) uplink random-coding argument uses a Slepian-Wolf/Cover distributed source code, e.g.~\cite[Sec.~10]{El-Gamal-2011-B}. The source sequence at each node is compressed to a \emph{bin} index that is sent to the Relay over the noiseless uplink channel. The $L$ compression rates of the source codes are selected to satisfy all but one inequality in the Slepian-Wolf Theorem~\cite[Thm.~10.3]{El-Gamal-2011-B} -- the exception being the total sum rate inequality, e.g.~\eqref{Eqn:HanExtraInequality} is omitted. The Relay has direct access to the $L$ bin indices; it does \emph{not} attempt to decode the individual source sequences. The (noisy) downlink random coding argument combines the virtual binning idea of Tuncel~\cite{Tuncel-Apr-2006-A} with the cascaded Slepian-Wolf binning idea of Wyner \emph{et al.}~\cite{Wyner-Jun-2002-A}. Each node will use a JSC-decoder to recover the source sequences.

\subsubsection{Preliminaries}

We use \emph{typical sequences} as, for example, defined in~\cite[Sec.~1]{Kramer-2008-A} and~\cite[Chap.~2]{El-Gamal-2011-B}. Suppose that $(A,B)$ are random variables on a discrete product space $\set{A} \times \set{B}$ with joint distribution $p_{AB}$. Let $\bm{\set{A}}$ and $\bm{\set{B}}$ denote the $n$-fold Cartesian products of $\set{A}$ and $\set{B}$ respectively. The \emph{type} of $\bm{a}$ in $\bm{\set{A}}$ and the \emph{joint type} of $(\bm{a},\bm{b})$ in $\bm{\set{A}} \times \bm{\set{B}}$ are the empirical distributions respectively defined by
\begin{align*}
\pi(a'| \bm{a}) \triangleq \frac{\big| \{i : a_i = a'\}\big|}{n},\ \ \ a' \in \set{A},
\end{align*}
and
\begin{equation*}
\pi( a',b' \big| \bm{a},\bm{b} ) \triangleq \frac{\big|\big\{i : (a_i,b_i) = (a',b') \big\} \big|}{n},\quad (a',b') \in \set{A} \times \set{B}.
\end{equation*}
Fix $\delta > 0$. The $\delta$-\emph{typical}, $\delta$-\emph{jointly-typical} and  $\delta$-\emph{conditionally-typical} sets are respectively defined by
\begin{equation*}
\set{T}_\delta(A) \triangleq \Big\{ \bm{a} \in \bm{\set{A}} : \Big|\pi(a' | \bm{a}) - p_A(a') \Big| \leq \delta p_A(a') ,\ \forall a' \in \set{A} \Big\},\\
\end{equation*}
\begin{align*}
\set{T}_\delta(A,B) \triangleq \Big\{ & (\bm{a},\bm{b}) \in \bm{\set{A}} \times \bm{\set{B}}  : \\
& \Big| \pi(a',b' | \bm{a},\bm{b}) - p_{AB}(a',b') \Big| \leq \delta p_{AB}(a',b'),\\
& \hspace{40mm} \forall (a',b') \in \set{A} \times \set{B} \Big\},
\end{align*}
and
\begin{equation*}
\set{T}_\delta(A,B|\bm{a}) \triangleq \Big\{\bm{b} \in \bm{\set{B}} : (\bm{a},\bm{b}) \in \set{T}_\delta(A,B) \Big\},
\end{equation*}
where $p_A$ denotes the $A$-marginal of $p_{AB}$. We note that if $(\bm{a},\bm{b})$ is in $\set{T}_{\delta}(A,B)$, then $\bm{a}$ is in $\set{T}_{\delta}(A)$ and $\bm{b}$ is in $\set{T}_{\delta}(B)$. The next lemma will be used throughout the proof.

\medskip

\begin{lemma}[Thms.~1.2 \& 1.3,~\cite{Kramer-2008-A}]
\label{Lem:CondtionallyTypical}
Fix
\begin{equation*}
0 \leq \delta_1 < \delta_2 \leq \min_{p_{AB}(a,b) \in \text{ support}(p_{AB})} p_{AB}(a,b).
\end{equation*}
If $\bm{a}$ belongs to $\set{T}_{\delta_1}(A)$, then the cardinality of the conditionally typical set $\set{T}_{\delta_2}(A,B|\bm{a})$ satisfies
\begin{equation*}
\big| \set{T}_{\delta_2}(A,B|\bm{a}) \big| \leq 2^{n H(B|A)(1+\delta_2)}.
\end{equation*}
Moreover, the probability that 
\begin{equation*}
\bm{B} \triangleq B_1,B_2,\ldots,B_n
\end{equation*}
(drawn i.i.d. with the $B$-marginal of $p_{AB}$) belongs to the conditionally-typical set satisfies
\begin{equation*}
\mathbb{P}\big[\bm{B} \in \set{T}_{\delta_2}(A,B|\bm{a}) \big] \leq 2^{-n(I(A;B) - 2 \delta_2 H(B))}.
\end{equation*}
\end{lemma}

\medskip

\subsubsection{Uplink Code Construction}\label{Sec:RandomCoding-Uplink}

Consider Node $l$. Randomly partition the source space $\bm{\set{W}}\bs{l}$ into $2^{n r\bs{l}}$ bins, labelled as $\{ \set{B}\bs{l}_1,\set{B}\bs{l}_2,\ldots,\set{B}\bs{l}_{2^{n r\bs{l}}}\}$, using an i.i.d. uniform law; i.e.,
\begin{equation*}
\mathbb{P}\big[ \bm{w}\bs{l} \in \set{B}\bs{l}_{b\bs{l}} \big] = \frac{1}{2^{n r\bs{l}}}, \quad b\bs{l} \in \big\{ 1,2,\ldots,2^{nr\bs{l}}\big\}.
\end{equation*}
The allowable values of $r\bs{l}$ will be specified later. With a slight abuse of notation, let
\begin{equation*}
\phi\bs{l} : \bm{w}\bs{l} \mapsto b\bs{l}
\end{equation*}
denote the map from source sequences to bin indices. Node~$l$ observes a sequence $\bm{w}\bs{l}$ from $\bm{\set{W}}\bs{l}$, and it sends the corresponding bin index $b\bs{l} \triangleq \phi\bs{l}(\bm{w}\bs{l})$ to the Relay.

\medskip

\subsubsection{Downlink Code}\label{Sec:RandomCoding-Downlink}

Consider the Cartesian product set of bin indices,
\begin{equation*}
\bm{\set{B}} \triangleq \bigotimes_{l=1}^L \big\{1,2,\ldots,2^{n r\bs{l}} \big\}.
\end{equation*}
For each bin tuple $\bm{b} = (b\bs{1},b\bs{2},\ldots,b\bs{L})$ in $\bm{\set{B}}$, generate a downlink codeword
\begin{equation*}
\bm{x}\bs{0}(\bm{b}) = \big(x\bs{0}_1(\bm{b}),x\bs{0}_2(\bm{b}),\ldots,x\bs{0}_{n}(\bm{b})\big)
\end{equation*}
by randomly drawing $n$-symbols from $\set{X}\bs{0}$ using the marginal distribution of $X\bs{0}$. The Relay observes a bin tuple $\bm{b} \triangleq \big(b\bs{1},b\bs{2},\ldots,b\bs{L} \big)$ from $\bm{\set{B}}$ on the uplink, and it sends $\bm{x}\bs{0}(\bm{b})$ over the downlink broadcast channel.

\medskip

\begin{figure*}
\begin{multline}\label{Eqn:Qset}
\set{Q}\bs{1}(\bm{w}\bs{1}) \triangleq \Big\{ (b\bs{2},b\bs{3},\ldots,b\bs{L}) \text{ for which there exists a unique tuple } (\tilde{\bm{w}}\bs{2},\tilde{\bm{w}}\bs{3},\ldots,\tilde{\bm{w}}\bs{L}) \\
\hspace{60mm} \text{ such that } \tilde{\bm{w}}\bs{l} \text{ belongs to } \set{B}\bs{l}_{b\bs{l}},\ \forall l = 2,3,\ldots,L, \text{ and }\\
(\tilde{\bm{w}}\bs{2},\tilde{\bm{w}}\bs{3},\ldots,\tilde{\bm{w}}\bs{L}) \in \set{T}_{\delta}(W\bs{1},\ldots,W\bs{L}|\bm{w}\bs{1}) \Big\}.
\end{multline}
\end{figure*}

\subsubsection{Decoding}\label{Sec:RandomCoding-Decoding}
Let $q_\text{s}$ denote the joint distribution of the source $(W\bs{1},W\bs{2},\ldots,W\bs{L})$, and let
\begin{equation*}
\bm{q}_\text{s}(\bm{w}\bs{1},\bm{w}\bs{2},\ldots,\bm{w}\bs{L}) \triangleq \prod_{i=1}^n q_\text{s}(w_i\bs{1},w_i\bs{2},\ldots,w_i\bs{L}).
\end{equation*}
Fix $0 < \delta \leq \mu$, where $\mu$ is the smallest value in the support sets of $q_\text{s}$ and the joint distribution $(X\bs{0},Y\bs{1})$.

The decoding procedure used at each node is identical; we describe the procedure for Node~1. The node observes a source sequence $\bm{w}\bs{1}$ from $\bm{\set{W}}\bs{1}$, it computes the corresponding bin index $b\bs{1} \equiv \phi\bs{1}(\bm{w}\bs{1})$, and it observes a channel output $\bm{y}\bs{1}$ from $\bm{\set{Y}}\bs{1}$. The node compiles a list of bin tuples that contain unique source sequences $\delta$-jointly typical with $\bm{w}\bs{1}$; to this end, for each $\bm{w}\bs{1}$ in $\bm{\set{W}}\bs{1}$ let $\set{Q}\bs{1}(\bm{w}\bs{1})$ be defined as in~\eqref{Eqn:Qset}.

In addition, the node looks for a \emph{unique} tuple $(\tilde{b}\bs{2},\tilde{b}\bs{3},$ $\ldots, \tilde{b}\bs{L})$ in $\set{Q}\bs{1}(\bm{w}\bs{1})$ such that the corresponding broadcast channel codeword is $\delta$-jointly typical with the observed channel output, i.e.,
\begin{equation}
\bm{x}\bs{0}(b\bs{1},\tilde{b}\bs{2}, \tilde{b}\bs{3}, \ldots, \tilde{b}\bs{L}) \in \set{T}_{\delta}(X\bs{0},Y\bs{1}|\bm{y}\bs{1}).
\end{equation}
If successful, Node~1 sets $(\bm{\hat{w}}\bs{1,1},\bm{\hat{w}}\bs{1,2},\bm{\hat{w}}\bs{1,3},\ldots,\bm{\hat{w}}\bs{1,L})$ equal to $(\bm{w}\bs{1},\bm{\tilde{w}}\bs{2},\bm{\tilde{w}}\bs{3},\ldots,\bm{\tilde{w}}\bs{L})$, where $(\bm{\tilde{w}}\bs{2},\bm{\tilde{w}}\bs{3},\ldots,$ $\bm{\tilde{w}}\bs{L})$ is the unique vector identified by $\set{Q}\bs{1}(\bm{w}\bs{1})$. If unsuccessful, the node declares an error.

\medskip

\emph{Remarks:}
\begin{itemize}
\item \emph{Noisy uplink:} At the end of the proof, we will adapt the uplink code for noiseless channels to include noisy channels by simply adding a good point-to-point channel code for each orthogonal uplink.
\item \emph{Separation:} Each node uses standalone source (and, later, channel) encoders on the uplink. Similarly, the Relay uses standalone channel decoders on the uplink and a standalone channel encoder on the downlink. However, the decoder at each node is a true JSC-decoder; for example, Node~$1$ first exploits its source side information $\bm{w}\bs{1}$ to compile the list $\set{Q}\bs{1}(\bm{w}\bs{1})$, before decoding the transmitted codeword from $\set{Q}\bs{1}(\bm{w}\bs{1})$ using $\bm{y}\bs{1}$. It can be suboptimal to decode the transmitted codeword $\bm{x}\bs{0}(b\bs{1},b\bs{2},\ldots,b\bs{L})$ using only the channel output $\bm{y}\bs{1}$, which, for example, is the case in separate source-channel coding.
\end{itemize}

\medskip

\subsubsection{Error Analysis}\label{Sec:RandomCoding-Analysis}

We wish to upper bound the average joint error probability $P_\text{e}$, as defined in~\eqref{Eqn:JointErrorProb}, for the described ensemble of codes. By the union bound for probability, we have
\begin{multline}\label{Eqn:AverageProbabilityOfError}
P_\text{e} \leq \sum_{l=1}^L \mathbb{P} \Big[ \big(\bm{\hat{W}}\bs{l,1},\bm{\hat{W}}\bs{l,2},\ldots,\bm{\hat{W}}\bs{l,L}\big)\\
 \neq (\bm{W}\bs{1},\bm{W}\bs{2},\ldots,\bm{W}\bs{L}\big)\Big].
\end{multline}
Consider the first error event of the sum~\eqref{Eqn:AverageProbabilityOfError}:
\begin{multline*}
\set{E} \triangleq \{ (\bm{\hat{W}}\bs{1,1},\bm{\hat{W}}\bs{1,2},\ldots,\bm{\hat{W}}\bs{1,L}) \\
\neq (\bm{W}\bs{1},\bm{W}\bs{2},\ldots,\bm{W}\bs{L})\};
\end{multline*}
i.e., the event that Node~1 decodes one or more source sequences in error. We now given an upper bound for $\mathbb{P}[\set{E}]$. Fix
\begin{equation*}
0 \leq \delta' < \delta
\end{equation*}
arbitrarily. Consider the following events.
\begin{enumerate}
\item[(i)] The event that the source sequences are not $\delta'$-jointly typical:
\begin{multline*}
\set{E}_1 \triangleq \big\{ \big(\bm{W}\bs{1},\bm{W}\bs{2},\ldots,\bm{W}\bs{L} \big) \\
\notin \set{T}_{\delta'}(W\bs{1},W\bs{2},\ldots,W\bs{L}) \big\}.
\end{multline*}
\item[(ii)] For each nonempty subset  
\begin{equation*}
\set{L} = \{l_1,l_2,\ldots,l_{|\set{L}|}\}
\end{equation*}
of $\{2,3,\ldots,L\}$, define the following event: the source sequences are $\delta'$-jointly typical and there exists an $|\set{L}|$-tuple of $\delta$-conditionally typical sequences, say
\begin{equation*}
\bm{\tilde{w}}\bs{\set{L}} = (\bm{w}\bs{l_1},\bm{w}\bs{l_2},\ldots,\bm{w}\bs{l_{|\set{L}|}}),
\end{equation*}
in the same bins as the source:
\begin{multline*}
\set{E}_{2,\set{L}} \triangleq \set{E}_1^c \cap
\Big\{ \exists \bm{\tilde{w}}\bs{\set{L}} \in \set{T}_{\delta}(W\bs{2},\ldots,W\bs{L}|\bm{W}\bs{1},\bm{W}\bs{\set{L}^c}) \\
\text{ such that } \bm{\tilde{w}}\bs{l} \neq \bm{W}\bs{l},\ \phi\bs{l}(\bm{\tilde{w}}\bs{l}) = \phi\bs{l}(\bm{W}\bs{l}), \ \forall l \in \set{L} \Big\}.
\end{multline*}
\item[(iii)] The event that the broadcast channel input and output at Node~1 are not $\delta'$-jointly typical:
\begin{multline*}
\set{E}_3 \triangleq \big\{ \big(\bm{X}\bs{0}(B\bs{1},B\bs{2},\ldots,B\bs{L}),\bm{Y}\bs{1} \big) \\
\notin \set{T}_{\delta'}(X\bs{0},Y\bs{1}) \big\}.
\end{multline*}
\item[(iv)] The source sequences are $\delta'$-jointly typical, the channel input and output are $\delta'$-jointly typical, and there exists another $\delta$-conditionally typical codeword with bin indices in $\set{Q}\bs{1}(\bm{W})$:
\begin{multline*}
\set{E}_4 \triangleq \set{E}_1^c \cap \set{E}_3^c \cap \Big\{ \exists (\tilde{b}\bs{2},\tilde{b}\bs{3},\ldots,\tilde{b}\bs{L})\in \set{Q}\bs{1}(\bm{W}\bs{1}) \\
\text{ such that } (\tilde{b}\bs{2},\tilde{b}\bs{3},\ldots,\tilde{b}\bs{L}) \neq (B\bs{2},B\bs{3},\ldots,B\bs{L})\\
\text{ and }  \bm{x}\bs{0}(B\bs{1},\tilde{b}\bs{2},\tilde{b}\bs{3},\ldots,\tilde{b}\bs{L}) \in \set{T}_{\delta}(X\bs{0},Y\bs{1}|\bm{Y}\bs{1})\Big\}.
\end{multline*}
\end{enumerate}
The error event $\set{E}$  --- the event that Node~1 decodes a source sequence in error --- is a subset of the union of $\set{E}_1$, $\cup \set{E}_{2,\set{L}}$, $\set{E}_3$ and $\set{E}_4$; hence,
\begin{equation}\label{Eqn:UnionBoundforE}
\mathbb{P}\big[\set{E}\big] \leq \mathbb{P}\big[\set{E}_1\big] + \sum_{\set{L} \subseteq \{2,3,\ldots,L\}}  \mathbb{P}\big[\set{E}_{2,\set{L}} \big] + \mathbb{P}\big[\set{E}_3\big] + \mathbb{P}\big[\set{E}_4\big].
\end{equation}
It follows from the \emph{law of large numbers} that
\begin{equation*}
\lim_{n\rightarrow\infty} \mathbb{P}\big[\set{E}_1\big] = 0\quad \text{ and }\quad \lim_{n\rightarrow\infty} \mathbb{P}\big[\set{E}_3\big] = 0
\end{equation*}
e.g., see~\cite[Sec.~2]{El-Gamal-2011-B} or~\cite[Thm.~1.1]{Kramer-2008-A}.

\begin{figure*}
\begin{align}
\notag
\mathbb{P}\big[\set{E}_{2,\set{L}} \big]
&\stackrel{}{=} \sum_{(\alpha)} \bm{q}_{\text{s}}(\bm{w}\bs{1},\bm{w}\bs{2},\ldots,\bm{w}\bs{L})
\indicator{(\bm{w}\bs{1},\bm{w}\bs{2},\ldots,\bm{w}\bs{L}) \in \set{T}_{\delta'}(W\bs{1},W\bs{2},\ldots,W\bs{L})}
\\
\notag
&\hspace{75mm} \cdot \sum_{(\beta)}
\mathbb{P}\big[ \phi\bs{l}(\bm{\tilde{w}}\bs{l}) = \phi\bs{l}(\bm{w}\bs{l})\ \forall l \in \set{L} \big]\\
\notag
&\stackrel{}{=} \sum_{(\gamma)} \bm{q}_{\text{s}}(\bm{w}\bs{1},\bm{w}\bs{2},\ldots,\bm{w}\bs{L})
\sum_{(\beta)}
\mathbb{P}\big[ \phi\bs{l}(\bm{\tilde{w}}\bs{l}) = \phi\bs{l}(\bm{w}\bs{l})\ \forall l \in \set{L} \big]\\
\notag
&\step{a}{=} \sum_{(\gamma)} \bm{q}_{\text{s}}(\bm{w}\bs{1},\bm{w}\bs{2},\ldots,\bm{w}\bs{L})
\sum_{(\beta)}
2^{-n \sum_{l \in \set{L}} r\bs{l}}\\
\notag
&\step{b}{\leq} \sum_{(\gamma)} \bm{q}_{\text{s}}(\bm{w}\bs{1},\bm{w}\bs{2},\ldots,\bm{w}\bs{L})\
2^{n H(W\bs{\set{L}}|W\bs{1},W\bs{\set{L}^c})(1+\delta)}\ 2^{-n \sum_{l \in \set{L}} r\bs{l}}\\
\label{Eqn:Bound2El}
&\leq 2^{n H(W\bs{\set{L}}|W\bs{1},W\bs{\set{L}^c})(1+\delta)} \ 2^{-n \sum_{l \in \set{L}} r\bs{l}},
\end{align}
\end{figure*}

The probability of each $\set{E}_{2,\set{L}}$ can be upper bound as shown in~\eqref{Eqn:Bound2El}, where
\begin{itemize}

\item $\indicator{\cdot}$ denotes the indicator function.

\item The sum marked with $(\alpha)$ is taken over all 
\begin{equation*}
(\bm{w}\bs{1},\bm{w}\bs{2},\ldots,\bm{w}\bs{L}) \in \bm{\set{W}}\bs{1} \times \bm{\set{W}}\bs{2} \times \cdots \times \bm{\set{W}}\bs{L}.
\end{equation*}

\item The sums marked with $(\beta)$ are taken over all 
\begin{multline*}
\bm{\tilde{w}}\bs{\set{L}} \in \set{T}_{\delta}(W\bs{1},W\bs{2},\ldots,W\bs{L}|\bm{w}\bs{1},\bm{w}\bs{\set{L}^c})  \\
\text{with}\quad  \phi(\bm{\tilde{w}}\bs{l}) \neq \phi(\bm{w}\bs{l})\quad \text{for all}\quad l \in \set{L}.
\end{multline*}

\item  the sums marked with $(\gamma)$ are take over all 
\begin{equation*}
(\bm{w}\bs{1},\bm{w}\bs{2},\ldots,\bm{w}\bs{L}) \in \set{T}_{\delta'}(W\bs{1},W\bs{2},\ldots,W\bs{L});
\end{equation*}

\item Equality (a) follows because the probability that each sequence $\bm{\tilde{\bm{w}}}\bs{l}$ in $\bm{\set{W}}\bs{l}$ is randomly assigned to the same bin as $\bm{w}\bs{l}$ is independent of all other bin assignments and equal to $2^{-n r\bs{l}}$.

\item Inequality (b) bounds the cardinality of of the conditionally typical set $\set{T}_{\delta}(W\bs{1},W\bs{2},\ldots,W\bs{L}|\bm{w}\bs{1},\bm{w}\bs{\set{L}^c})$ using Lemma~\ref{Lem:CondtionallyTypical}.

\end{itemize}
Finally, we notice that if
\begin{equation}\label{Eqn:Node1UplinkRateConstraint}
H(W\bs{\set{L}}|W\bs{1},W\bs{\set{L}^c}) (1 + \delta) < \sum_{l\in\set{L}} r\bs{l}
\end{equation}
then 
\begin{equation*}
\lim_{n \rightarrow \infty} \mathbb{P}[\set{E}_{2,\set{L}}] = 0.
\end{equation*}

\begin{figure*}
\begin{align}
\notag
\mathbb{P}[\set{E}_4] &\step{a}{=}
\sum_{(\alpha)} \bm{q}_\text{s}(\bm{w}\bs{1},\bm{w}\bs{2},\ldots,\bm{w}\bs{L})
\indicator{(\bm{w}\bs{1},\bm{w}\bs{2},\ldots,\bm{w}\bs{L}) \in \set{T}_{\delta'}(W\bs{1},W\bs{2},\ldots,W\bs{L})}\\
\notag
&\hspace{10mm}
\cdot
\sum_{\bm{y}\bs{1} \in \bm{\set{Y}}\bs{1}} \mathbb{P}\big[\bm{Y}\bs{1}=\bm{y}\bs{1}\big|(\bm{W}\bs{1},\bm{W}\bs{2},\ldots,\bm{W}\bs{L}) = (\bm{w}\bs{1},\bm{w}\bs{2},\ldots,\bm{w}\bs{L})\big]\\
\notag
&\hspace{40mm}  \cdot \indicator{\big(\bm{X}\bs{0}(b\bs{1},\ldots,b\bs{L}),\bm{y}\bs{1}\big) \in \set{T}_{\delta'}(X\bs{0},Y\bs{1})}\\
\notag
&\hspace{59mm}
\cdot \sum_{(\beta)} \mathbb{P}\big[\bm{X}\bs{0}(b\bs{1},\tilde{b}\bs{2},\ldots,\tilde{b}\bs{L}) \in \set{T}_\delta(X\bs{0},Y\bs{1}|\bm{y}\bs{1})\big]\\
\notag
&\step{b}{\leq}
\sum_{(\gamma)} \bm{q}_\text{s}(\bm{w}\bs{1},\bm{w}\bs{2},\ldots,\bm{w}\bs{L})\\
\notag
&\hspace{10mm}\cdot \sum_{\bm{y}\bs{1} \in \set{T}_{\delta'}(Y\bs{1})} \mathbb{P}\big[\bm{Y}\bs{1}=\bm{y}\bs{1}\big|(\bm{W}\bs{1},\bm{W}\bs{2},\ldots,\bm{W}\bs{L}) = (\bm{w}\bs{1},\bm{w}\bs{2},\ldots,\bm{w}\bs{L})\big] \\
\notag
&\hspace{59mm}
\cdot \sum_{(\beta)} \mathbb{P}\big[\bm{X}\bs{0}(b\bs{1},\tilde{b}\bs{2},\ldots,\tilde{b}\bs{L}) \in \set{T}_\delta(X\bs{0},Y\bs{1}|\bm{y}\bs{1})\big]\\
\notag
&\step{c}{\leq}
\sum_{(\gamma)} \bm{q}_\text{s}(\bm{w}\bs{1},\bm{w}\bs{2},\ldots,\bm{w}\bs{L})\\
\notag
&\hspace{10mm}
\sum_{\bm{y}\bs{1} \in \set{T}_{\delta'}(Y\bs{1})} \mathbb{P}\big[\bm{Y}\bs{1} = \bm{y}\bs{1} \big| (\bm{W}\bs{1},\bm{W}\bs{2},\ldots,\bm{W}\bs{L})=(\bm{w}\bs{1},\bm{w}\bs{2},\ldots,\bm{w}\bs{L})\big]
 \\
\notag
&\hspace{80mm} \cdot |\set{Q}\bs{1}(\bm{w}\bs{1})|\ 2^{-n(I(X\bs{0};Y\bs{1})-2\delta H(X\bs{0}))}\\
\notag
\\
\notag
&\step{d}{\leq}
\sum_{(\gamma)} \bm{q}_\text{s}(\bm{w}\bs{1},\bm{w}\bs{2},\ldots,\bm{w}\bs{L})\\
\notag
&\hspace{10mm}
\sum_{\bm{y}\bs{1} \in \set{T}_{\delta'}(Y\bs{1})}
\mathbb{P}\big[\bm{Y}\bs{1}=\bm{y}\bs{1}\big|(\bm{W}\bs{1},\bm{W}\bs{2},\ldots,\bm{W}\bs{L})
= (\bm{w}\bs{1},\bm{w}\bs{2},\ldots,\bm{w}\bs{L})\big]\\
\notag
&\hspace{50mm} \cdot 2^{n H(W\bs{2},W\bs{3},\ldots,W\bs{L}|W\bs{1})(1 + \delta)}\ 2^{-n(I(X\bs{0};Y\bs{1})-2\delta H(X\bs{0}))}\\
\label{Eqn:BoundE4l}
&\leq 2^{n H(W\bs{2},W\bs{3},\ldots,W\bs{L}|W\bs{1})(1 + \delta)}\ 2^{-n(I(X\bs{0};Y\bs{1})-2\delta H(X\bs{0}))},
\end{align}
\end{figure*}
The probability of $\set{E}_4$ can be upper bound as shown in~\eqref{Eqn:BoundE4l}.
\begin{itemize}

\item The sum marked with $(\alpha)$ is taken over all 
\begin{equation*}
(\bm{w}\bs{1},\bm{w}\bs{2},\ldots,\bm{w}\bs{L}) \in \bm{\set{W}}\bs{1} \times \bm{\set{W}}\bs{2} \times \cdots \times \bm{\set{W}}\bs{L}.
\end{equation*}

\item The sums marked with $(\beta)$ are taken over all 
\begin{equation*}
(\tilde{b}\bs{2},\tilde{b}\bs{3},\ldots,\tilde{b}\bs{L}) \in \set{Q}\bs{1}(\bm{w}\bs{1}).
\end{equation*}

\item The sums marked with $(\gamma)$ are taken over all 
\begin{equation*}
(\bm{w}\bs{1},\bm{w}\bs{2},\ldots,\bm{w}\bs{L})  \in \set{T}_\delta(W\bs{1},W\bs{2},\ldots,W\bs{L}).
\end{equation*}

\end{itemize}
The reasoning behind (in)equalities (a) through (d) is as follows.
\begin{itemize}
\item[(a)] We let 
\begin{equation*}
b\bs{l} = \phi\bs{l}(\bm{w}\bs{l}),
\end{equation*}
for $l = 1,2,\ldots,L$, denote the bin index of the $l$-th source sequence.
\item[(b)] We have that 
\begin{equation*}
\big(\bm{X}\bs{0}(b\bs{1},\ldots,b\bs{L}),\bm{y}\bs{1}\big) \in \set{T}_{\delta'}(X\bs{0},Y\bs{1})
\end{equation*}
implies 
\begin{equation*}
\bm{y}\bs{1} \in \set{T}_{\delta'}(Y\bs{1}).
\end{equation*}
\item[(c)] We bound the probability that an alternate codeword, $\bm{X}\bs{0}(b\bs{1},\tilde{b}\bs{2},\tilde{b}\bs{2},\ldots,\tilde{b}\bs{L})$, is $\delta$-jointly typical with $\bm{y}\bs{1}$, using Lemma~\ref{Lem:CondtionallyTypical} and the fact that the $n$-symbols of the codeword are drawn i.i.d. with the marginal of $X\bs{0}$.
\item[(d)] The cardinality of $\set{Q}\bs{1}(\bm{w}\bs{1})$ must be smaller than the cardinality of the conditionally typical set $\set{T}_{\delta}(W\bs{1},W\bs{2},$ $\ldots,W\bs{L}|\bm{w}\bs{1})$, which in turn is smaller than the bound of Lemma~\ref{Lem:CondtionallyTypical}.
\end{itemize}
Finally, it follows that 
\begin{equation*}
\lim_{n \rightarrow \infty} \mathbb{P}[\set{E}_4] = 0
\end{equation*}
whenever
\begin{multline}\label{Eqn:Node1DownlinkRateConstraint}
H(W\bs{2},W\bs{3},\ldots,W\bs{L}|W\bs{1}) < I(X\bs{0};Y\bs{1}) \\
- \delta\big(H(W\bs{2},W\bs{3},\ldots,W\bs{L}|W\bs{1}) + 2 H(X\bs{0})\big).
\end{multline}

The above analysis can be repeated for each of the $L$ nodes to obtain constraints analogous to~\eqref{Eqn:Node1UplinkRateConstraint} and~\eqref{Eqn:Node1DownlinkRateConstraint}. We are free to choose $\delta'$ and $\delta$ arbitrarily small, so $P_\text{e}$ (averaged over the ensemble of codes) can be made arbitrarily small by increasing $n$ if
\begin{equation}\label{Eqn:UplinkRateConstraint}
H(W\bs{\set{L}}|W\bs{\set{L}^c}) < \sum_{l\in\set{L}} r\bs{l},
\end{equation}
holds for each nonempty and strict subset $\set{L}$ of $\{1,2,\ldots,L\}$ and
\begin{equation}\label{Eqn:DownlinkRateConstraint}
H(W\bs{2},W\bs{3},\ldots,W\bs{L}|W\bs{l}) < I(X\bs{0};Y\bs{l}),
\end{equation}
holds for each $l = 1,2,\ldots,L$. The random coding argument implies that there must exist at least one code from the ensemble with an error probability at least as small as $P_\text{e}$.

\subsubsection{Extension to Noisy Uplink Channels}

Consider the setup of Theorem~\ref{Thm:JSCC} with noisy (orthogonal) uplink channels. The random-coding argument of Sections~\ref{Sec:RandomCoding-Uplink} to~\ref{Sec:RandomCoding-Decoding} can be extended to this setting by the use of a separate source and channel code architecture on the uplink. That is, use a (point-to-point) capacity-approaching code for each orthogonal uplink channel, choose the source-code compression rates to match the channel coding rates ($r\bs{l} = \Cu{l}-\zeta$, where $\zeta$ is sufficiently small), and communicate the $L$ bin indices to the relay using the point-to-point channel codes. \hfill $\blacksquare$

\subsection{Proof of Theorem~\ref{Thm:JSCC} (Converse)}
Suppose that we have a JSC-code with $P_\text{e} \leq \epsilon$. The next two lemmas will be useful.

\medskip

\begin{lemma}\label{Lem:Converse-Fano}
For each nonempty and strict subset $\set{L}$ of $\{1,2,\ldots,L\}$, we have
\begin{align*}
\frac{1}{n} H(\bm{W}\bs{\set{L}}|\bm{W}\bs{\set{L}^c},\bm{Y}\bs{\set{L}^c})
&\leq \varepsilon(n,\epsilon),
\end{align*}
where $\varepsilon(n,\epsilon) \rightarrow 0$ as $\epsilon \rightarrow 0$
\end{lemma}

\medskip

\begin{proof}
Choose $l$ from $\set{L}^c$ arbitrarily. We have
\begin{align*}
\frac{1}{n} H(\bm{W}\bs{\set{L}}|\bm{W}\bs{\set{L}^c},\bm{Y}\bs{\set{L}^c})
&\leq \frac{1}{n} \sum_{l \in \set{L}} H(W\bs{l} | \bm{W}\bs{j},\bm{Y}\bs{j})\\
&\step{a}{\leq} \frac{1}{n} \sum_{l \in \set{L}} \Big[ h(\epsilon) + \epsilon n \log |\set{W}\bs{l}| \Big]\\
&\leq \frac{L\ h(\epsilon)}{n}   + \epsilon L \max_{l' \in \{1,\ldots,L\}} |\set{W}\bs{l'}|,
\end{align*}
where $h(\epsilon)$ is the binary entropy function and (a) follows from Fano's inequality.
\end{proof}

\medskip

\begin{lemma}\label{Lem:Converse-Orthogonal-MAC}
For each nonempty and strict subset $\set{L}$ of $\{1,2,$ $\ldots,L\}$, we have
\begin{equation*}
\sum_{l \in \set{L}} I(\bm{X}\bs{l} ; \bm{Y}\bs{0,l}) \geq I(\bm{X}\bs{\set{L}} ; \bm{Y}\bs{0,\set{L}}).
\end{equation*}
\end{lemma}

\medskip

\begin{proof}
We have
\begin{align*}
\sum_{l \in \set{L}}  I(\bm{X}\bs{l} ; \bm{Y}\bs{0,l})
&= \sum_{l \in \set{L}} \Big[H(\bm{Y}\bs{0,l}) - H(\bm{Y}\bs{0,l}|\bm{X}\bs{l}) \Big]\\
&\geq H(\bm{Y}\bs{0,\set{L}}) - \sum_{l \in \set{L}} H(\bm{Y}\bs{0,l} | \bm{X}\bs{l})\\
&\step{a}{\geq} H(\bm{Y}\bs{0,\set{L}}) - H(\bm{Y}\bs{0,\set{L}}|\bm{X}\bs{\set{L}})\\
&= I(\bm{X}\bs{\set{L}} ; \bm{Y}\bs{0,\set{L}}),
\end{align*}
where (a) is a consequence of the Markov chain
\begin{equation*}
\bm{Y}\bs{0,l} \markov\ \bm{X}\bs{l} \markov\ \bm{Y}\bs{\set{L}\backslash l}.
\end{equation*}
\end{proof}

\medskip

\begin{proof}[Proof of Theorem~\ref{Thm:JSCC} (Converse)]
For each nonempty and strict subset $\set{L}$ of $\{1,2,\ldots,L\}$, we can lower bound the right hand side of~\eqref{Eqn:JSCCUp} by
\begin{align}
\notag
\sum_{l \in \set{L}} \Cu{l}
\notag
&\stackrel{(a)}{\geq} \sum_{l \in \set{L}} \frac{1}{n} I(\bm{X}\bs{l};\bm{Y}\bs{0,l})\\
\notag
&\stackrel{(b)}{\geq} \frac{1}{n} I(\bm{X}\bs{\set{L}};\bm{Y}\bs{0,\set{L}})\\
\notag
&\stackrel{(c)}{\geq} \frac{1}{n} I(\bm{W}\bs{\set{L}},\bm{W}\bs{\set{L}^c} ; \bm{Y}\bs{0,\set{L}})\\
\notag
&\geq \frac{1}{n} I(\bm{W}\bs{\set{L}} ; \bm{Y}\bs{0,\set{L}} | \bm{W}\bs{\set{L}^c})\\
\notag
&\stackrel{(d)}{\geq} \frac{1}{n} I(\bm{W}\bs{\set{L}} ; \bm{Y}\bs{\set{L}^c} | \bm{W}\bs{\set{L}^c})\\
\notag
&\geq \frac{1}{n} \Big[ H(\bm{W}\bs{\set{L}} | \bm{W}\bs{\set{L}^c}) - H( \bm{W}\bs{\set{L}} | \bm{Y}\bs{\set{L}^c},\bm{W}\bs{\set{L}^c}) \Big]\\
\label{Eqn:Converse1}
&\stackrel{(e)}{\geq} H(W\bs{\set{L}} | W\bs{\set{L}^c}) - \varepsilon(n, \epsilon),
\end{align}
where (a) follows from the definition of channel capacity; (b) follows from Lemma~\ref{Lem:Converse-Orthogonal-MAC}; (c) follows from the Markov chain $(\bm{W}\bs{\set{L}},\bm{W}\bs{\set{L}^c}) \markov \bm{X}\bs{\set{L}} \markov \bm{Y}\bs{0,\set{L}}$; (d) follows from the Markov chain $\bm{W}\bs{\set{L}} \markov (\bm{W}\bs{\set{L}^c},\bm{Y}\bs{0,\set{L}}) \markov \bm{Y}\bs{\set{L}^c}$; and (e) follows from Lemma~\ref{Lem:Converse-Fano} and that the source is i.i.d.

Consider~\eqref{Eqn:JSCCDown}. Let $p_{X_i\bs{0}}$ denote the pmf of the $i$-th symbol, $X\bs{0}_i$, of $\bm{X}\bs{0}$. Define a new time-averaged random variable $\tilde{X}\bs{0}$ on $\set{X}\bs{0}$ via the pmf
\begin{equation*}
p_{\tilde{X}\bs{0}}(x) \triangleq \frac{1}{n} \sum_{i = 1}^{n} p_{X\bs{0}_i}(x), \quad x \in \set{X}\bs{0}.
\end{equation*}
For each $l$ in $\{1,2,\ldots,L\}$, we have
\begin{align}
\notag
I(\tilde{X}\bs{0} ; Y\bs{l})
&\stackrel{(a)}{\geq} \frac{1}{n} \sum_{i=1}^{n} I(X\bs{0}_i ; Y\bs{l}_i)\\
\notag
&\stackrel{(b)}{\geq} \frac{1}{n} I(\bm{X}\bs{0} ; \bm{Y}\bs{l})\\
\notag
&\stackrel{(c)}{\geq} \frac{1}{n} I(\bm{W}\bs{l},\bm{W}\bs{\{l\}^c} ; \bm{Y}\bs{l})\\
\notag
&\geq \frac{1}{n} I(\bm{W}\bs{\{l\}^c} ; \bm{Y}\bs{l}|\bm{W}\bs{l})\\
\label{Eqn:Converse2}
&\stackrel{(d)}{\geq} H(W\bs{\{l\}^c} | W\bs{l}) -  \varepsilon(n , \epsilon),
\end{align}
where (a) follows from Jensen's inequality and the fact that $I(\tilde{X}\bs{0};Y\bs{l})$ is concave in $p_{\tilde{X}\bs{0}}$; (b) follows because the broadcast channel is stationary and memoryless; (c) follows from the Markov chain
$(\bm{W}\bs{l},\bm{W}\bs{\{l\}^c}) \markov \bm{X}\bs{0} \markov \bm{Y}\bs{l}$; and (d) follows from Lemma~\ref{Lem:Converse-Fano}.

Consider a sequence $\{\epsilon\} \rightarrow 0$. For each $\epsilon$ in this sequence, we have by definition at JSC-code for which~\eqref{Eqn:Converse1} and~\eqref{Eqn:Converse2} hold for some $X\bs{0}$ on $\set{X}\bs{0}$. The proof is completed by noting that the sequence of pmfs $p_{\tilde{X}\bs{0}_i}$ will converge to some pmf on $\set{X}\bs{0}$.
\end{proof}

\section{Counterexample}\label{App:Example}

We now describe a situation where the achievability assertion of Theorem~\ref{Thm:JSCC} holds, but that of Theorem~\ref{Thm:SSCC} fails; equivalently, reliable communication is achievable with JSC-codes of the form~\eqref{Eqn:DefJSCC}, and it is not achievable with separate source and channel codes of the form~\eqref{Eqn:ChannelCode} and~\eqref{Eqn:SourceCodeDefinition}.

\medskip

\begin{example}
Suppose we have $3$-nodes with the following sources:
\begin{align}
W\bs{1} &= U_{1} \oplus U_{12} \oplus U_{13} \\
W\bs{2} &= U_{2} \oplus U_{12} \oplus U_{23} \\
W\bs{3} &= U_{3} \oplus U_{13} \oplus U_{23},
\end{align}
where each $U_{i} \in \{0,1\}$ is an independent random variable, for all $i \in \{1,2,3,12,13,23\}$, and $\oplus$ denotes the XOR function.  We choose $\Pr\{ U_{1} = 1\}=0.0085$, $\Pr\{ U_{2} = 1\}=\Pr\{ U_{3} = 1\}=0.0052$, $\Pr\{ U_{12} = 1\} = \Pr\{ U_{13} = 1\} = 0.0128$, and $\Pr\{ U_{23} = 1\}=0.138$. For these choices of probability mass functions, we have 
\begin{subequations}\label{Eqn:CounterexampleSourceEntropies}
\begin{align}
\notag
H(W\bs{1}|W\bs{2},W\bs{3}) &= H(W\bs{2}|W\bs{1},W\bs{3})\\
\notag
& = H(W\bs{3}|W\bs{1},W\bs{2})\\
& = 0.10\\
\notag
H(W\bs{1},W\bs{2}|W\bs{3}) &= H(W\bs{1},W\bs{3}|W\bs{2}) \\
&= 0.30\\
H(W\bs{2},W\bs{3}|W\bs{1}) &= 0.70 \\
H(W\bs{1}) &= 0.21.
\end{align}
\end{subequations}
The {\em information diagram}~\cite[Sec.\ 3.5]{Yeung-2008-B} of such a source is depicted in Fig.~\ref{fig}. 

Suppose the uplink channels are point-to-point channels with capacities $\Cu{1}$ $=$ $\Cu{2}$ $=$ $\Cu{3}$ $=$ $2$, and the downlink channels are $Y\bs{i} = X\bs{0} \oplus N\bs{i}$, for each $i \in \{1,2,3\}$, where each $X\bs{i} \in \{0,1\}$, and each $N\bs{i}$ is an independent random variable. We choose $\Pr\{N\bs{1}=1\} = 0.0508$ and $\Pr\{N\bs{2}=1\} = \Pr\{ N\bs{3} = 1 \} = 0.184$. Note that the uniform input distribution simultaneously maximises $I(X\bs{0};Y\bs{i})$ for all $i \in \{1,2,3\}$, giving $I(X\bs{0};Y\bs{1}) = 0.71$ and $I(X\bs{0};Y\bs{2}) = I(X\bs{0};Y\bs{3}) = 0.31$. These assumptions satisfy the achievability requirements of Theorem~\ref{Thm:JSCC}.

We now show, via a contradiction, that there \emph{does not} exists a rate tuple $(r\bs{1},r\bs{2},r\bs{3})$ satisfying~\eqref{Eqn:Separation}. Suppose there exists $(r\bs{1},r\bs{2},r\bs{3})$ such that
\begin{subequations}
\begin{align}
0.10 =  H(W\bs{1}|W\bs{2},W\bs{3}) <\ & r\bs{1} < \Cu{1} = 2 \label{eq:1}\\
0.10 =  H(W\bs{2}|W\bs{1},W\bs{3}) <\ & r\bs{2} < \Cu{2} = 2 \\
0.10 =  H(W\bs{3}|W\bs{1},W\bs{2}) <\ & r\bs{3} < \Cu{3} = 2
\end{align}

\begin{align}
%\notag
%
0.30 = H(W\bs{1},W\bs{2}|W\bs{3}) <\  r\bs{1} + r\bs{2} 
&< \Cu{1} + \Cu{2} = 4 \\
%
%\notag
0.70 = H(W\bs{2},W\bs{3}|W\bs{1}) <\ r\bs{2} + r\bs{3} 
&< \Cu{2} + \Cu{3} = 4 \label{eq:2}\\
%
%\notag
0.30 = H(W\bs{1},W\bs{3}|W\bs{2}) <\ r\bs{1} + r\bs{3}
& < \Cu{1} + \Cu{3} = 4
\end{align}
and
\begin{align}
&r\bs{1} + r\bs{2} < I(X\bs{0};Y\bs{3}) = 0.31 \label{eq:3}\\
&r\bs{2} + r\bs{3} < I(X\bs{0};Y\bs{1}) = 0.71\\
& r\bs{1} + r\bs{3} < I(X\bs{0};Y\bs{2}) = 0.31. \label{eq:4}
\end{align}
\end{subequations}
The left inequalities in \eqref{eq:1} and \eqref{eq:2} imply that
\begin{equation*}
\max \{ r\bs{1} + r\bs{2} , r\bs{1} + r\bs{3} \} > 0.45.
\end{equation*}
However, \eqref{eq:3} and~\eqref{eq:4} together imply
\begin{equation*}
\max \{ r\bs{1} + r\bs{2}, r\bs{1} + r\bs{3} \} < 0.31,
\end{equation*}
which gives the desired contradiction.
\end{example}

\begin{figure}[t]
\centering
\includegraphics[width=4cm]{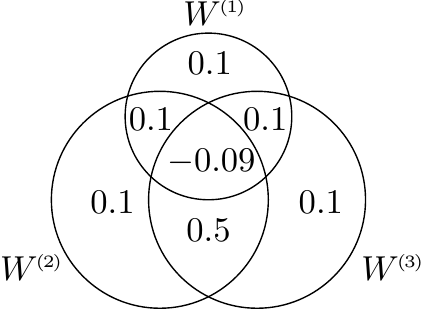}
\caption{The information diagram of $(W\bs{1},W\bs{2},W\bs{3})$ defined in~\eqref{Eqn:CounterexampleSourceEntropies}. Each of the three circles represents a source message, and the intersection between two circles represent the mutual information between the sources. The numbers are the values of the corresponding conditional entropies and mutual informations.}
\label{fig}
\end{figure}

%%%%
%%%% Bibliography
%%%%

\bibliographystyle{IEEEtran}
%\bibliography{Library}

\begin{thebibliography}{10}
\providecommand{\url}[1]{#1}
\csname url@samestyle\endcsname
\providecommand{\newblock}{\relax}
\providecommand{\bibinfo}[2]{#2}
\providecommand{\BIBentrySTDinterwordspacing}{\spaceskip=0pt\relax}
\providecommand{\BIBentryALTinterwordstretchfactor}{4}
\providecommand{\BIBentryALTinterwordspacing}{\spaceskip=\fontdimen2\font plus
\BIBentryALTinterwordstretchfactor\fontdimen3\font minus
  \fontdimen4\font\relax}
\providecommand{\BIBforeignlanguage}[2]{{%
\expandafter\ifx\csname l@#1\endcsname\relax
\typeout{** WARNING: IEEEtran.bst: No hyphenation pattern has been}%
\typeout{** loaded for the language `#1'. Using the pattern for}%
\typeout{** the default language instead.}%
\else
\language=\csname l@#1\endcsname
\fi
#2}}
\providecommand{\BIBdecl}{\relax}
\BIBdecl

\bibitem{Wyner-Jun-2002-A}
A.~Wyner, J.~Wolf, and F.~Willems, ``Communicating via a processing broadcast
  satellite,'' \emph{IEEE Transactions on Information Theory}, vol.~48, no.~6,
  pp. 1243--1249, 2002.

\bibitem{Rankov-Jul-2006-C}
B.~Rankov and A.~Wittneben, ``Achievable rate regions for the two-way relay
  channel,'' in \emph{proceedings IEEE International Symposium on Information
  Theory}, Seattle, WA, 2006.

\bibitem{Knopp-Feb-2006-C}
R.~Knopp, ``Two-way radio networks with a star topology,'' in \emph{Zurich
  Seminar on Communications}, Zurich, Switzerland, 2006.

\bibitem{Rankov-Feb-2007-A}
B.~Rankov and A.~Wittneben, ``Spectral efficient protocols for half-duplex
  fading relay channels,'' \emph{IEEE Journal on Selected Areas in
  Communications}, vol.~25, no.~2, pp. 379--389, 2007.

\bibitem{Nam-Nov-2010-A}
W.~Nam, S.-Y. Chung, and Y.~H. Lee, ``Capacity of the {G}aussian two-way relay
  channel to within 1/2 bit,'' \emph{IEEE Transactions on Information Theory},
  vol.~56, no.~11, pp. 5488--5494, 2010.

\bibitem{Cui-2012-A}
T.~Cui, J.~Kliewer, and T.~Ho, ``Communication protocols for {$N$}-way all-cast
  relay networks,'' \emph{IEEE Transactions on Communications}, vol.~PP,
  no.~99, pp. 1--13, 2012.

\bibitem{Cui-Oct-2009-A}
T.~Cui, T.~Ho, and J.~Kliewer, ``Memoryless relay strategies for two-way relay
  channels,'' \emph{IEEE Transactions on Communications}, vol.~57, no.~10, pp.
  3132--3143, 2009.

\bibitem{Katti-Jun-2008-A}
S.~Katti, H.~Rahul, D.~Katabi, M.~Medard, and J.~Crowcroft, ``{XOR}s in the
  air: practical wireless network coding,'' \emph{IEEE/ACM Transactions on
  Networking}, vol.~16, no.~3, pp. 497--510, 2008.

\bibitem{Jindal-Nov-2008-A}
A.~Jindal and K.~Psounis, ``Modeling spatially correlated data in sensor
  networks,'' \emph{ACM Transactions on Sensor Networks}, vol.~2, no.~4, pp.
  466--499, 2006.

\bibitem{Cover-2006-B}
T.~Cover and J.~Thomas, \emph{Elements of information theory}.\hskip 1em plus
  0.5em minus 0.4em\relax John Wiley and Sons, 2006.

\bibitem{Gallager-1968-B}
R.~Gallager, \emph{Information theory and reliable communication}.\hskip 1em
  plus 0.5em minus 0.4em\relax John Wiley and Sons, Inc. New York, NY, USA,
  1968.

\bibitem{Ahlswede-May-1983-A}
R.~Ahlswede and T.~Han, ``On source coding with side information via a
  multiple-access channel and related problems in multi-user information
  theory,'' \emph{IEEE Transactions on Information Theory}, vol.~29, no.~3, pp.
  396--412, 1983.

\bibitem{Yeung-2008-B}
R.~W. Yeung, \emph{Information theory and network coding}.\hskip 1em plus 0.5em
  minus 0.4em\relax Springer, 2008.

\bibitem{El-Gamal-2011-B}
A.~El~Gamal and Y.-H. Kim, \emph{Network {I}nformation {T}heory}.\hskip 1em
  plus 0.5em minus 0.4em\relax Cambridge University Press, 2011.

\bibitem{Han-1980-A}
T.~S. Han, ``Slepian-{W}olf-{C}over theorem for networks of channels,''
  \emph{Information and Control}, vol.~47, no.~1, pp. 67--83, 1980.

\bibitem{Tuncel-Apr-2006-A}
E.~Tuncel, ``Slepian-{W}olf coding over broadcast channels,'' \emph{IEEE
  Transactions on Information Theory}, vol.~52, no.~4, pp. 1469--1482, 2006.

\bibitem{Kschischang-Feb-2001-A}
F.~R. Kschischang, B.~J. Frey, and H.~A. Loeliger, ``Factor graphs and the
  sum-product algorithm,'' \emph{IEEE Transactions on Information Theory},
  vol.~47, no.~2, 2001.

\bibitem{Slepian-Jul-1973-A}
D.~Slepian and J.~Wolf, ``Noiseless coding of correlated information sources,''
  \emph{IEEE Transactions on information Theory}, vol.~19, no.~4, pp. 471--480,
  1973.

\bibitem{Barros-Jan-2006-A}
J.~Barros and S.~Servetto, ``Network information flow with correlated
  sources,'' \emph{IEEE Transactions on Information Theory}, vol.~52, no.~1,
  pp. 155--170, 2006.

\bibitem{Kramer-2008-A}
G.~Kramer, ``Topics in multi-user information theory,'' \emph{Foundations and
  Trends in Communications and Information Theory}, vol.~4, no. 4–5, pp.
  265--444, 2008.

\bibitem{Cover-Mar-1975-A}
T.~Cover, ``A proof of the data compression theorem of {S}lepian and {W}olf for
  ergodic sources,'' \emph{IEEE Transactions on Information Theory}, vol.~21,
  no.~2, pp. 226--228, 1975.

\bibitem{Oechtering-Jul-2012-C}
T.~Oechtering, M.~Wigger, and R.~Timo, ``Broadcast capacity regions with three
  receivers and message cognition,'' in \emph{in proceedings IEEE International
  Symposium on Information Theory}, Cambridge, MA, 2012.

\bibitem{Garcia-Frias-Sep-2001-A}
J.~Garcia-Frias and J.~D. Villasenor, ``Joint turbo decoding and estimation of
  hidden {M}arkov sources,'' \emph{IEEE Journal on Selected Areas in
  Communications}, vol.~19, no.~9, pp. 1671--1679, 2001.

\bibitem{Xu-May-2007-A}
X.~Qian, V.~Stankovic, and X.~Zixiang, ``Distributed joint source-channel
  coding of video using raptor codes,'' \emph{IEEE Journal on Selected Areas in
  Communications}, vol.~25, no.~4, pp. 851--861, 2007.

\bibitem{Gallager-Jan-1962-A}
R.~Gallager, ``Low-density parity-check codes,'' \emph{IRE Transactions on
  Information Theory}, vol.~8, no.~1, pp. 21--28, 1962.

\bibitem{MacKay-Mar-1999-A}
D.~J.~C. MacKay, ``Good error-correcting codes based on very sparse matrices,''
  \emph{IEEE Transactions on Information Theory}, vol.~45, no.~2, pp. 399--431,
  1999.

\bibitem{Johnson-2010-B}
S.~J. Johnson, \emph{Iterative error correction: turbo, low-density
  parity-check and repeat-accumulate codes}.\hskip 1em plus 0.5em minus
  0.4em\relax Cambridge University Press, 2010.

\bibitem{Richardson-2008-B}
T.~Richardson and R.~Urbanke, \emph{Modern coding theory}.\hskip 1em plus 0.5em
  minus 0.4em\relax Cambridge University Press, 2008.

\bibitem{Rathi-Nov-2009-C}
V.~Rathi, M.~Andersson, R.~Thobaben, J.~Kliewer, and M.~Skoglund, ``Two edge
  type {LDPC} codes for the wiretap channel,'' in \emph{in proceedings Asilomar
  Conference on Signals, Systems and Computers}, Pacific Grove, CA, 2009.

\bibitem{Azmi-Nov-2011-A}
M.~H. Azmi, Y.~Jinhong, G.~Lechner, and L.~K. Rasmussen, ``Design of
  multi-edge-type bilayer-expurgated {LDPC} codes for decode-and-forward in
  relay channels,'' \emph{IEEE Transactions on Communications}, vol.~59,
  no.~11, pp. 2993--3006, 2011.

\bibitem{Peleg-Nov-2006-C}
M.~Peleg, A.~Sanderovich, and S.~Shamai, ``On extrinsic information of good
  codes operating over memoryless channels with incremental noisiness,'' in
  \emph{IEEE Conv. Electrical Electronics Engineers Israel}, 2006.

\bibitem{Wyner-Jan-1974-A}
A.~Wyner, ``Recent results in the {S}hannon {T}heory,'' \emph{IEEE Transactions
  on Information Theory}, vol.~20, no.~1, pp. 2--10, 1974.

\bibitem{Chen-Sep-2009-A}
C.~Chen, D.-k. He, and A.~Jagmohan, ``The equivalence between {S}lepian-{W}olf
  coding and channel coding under density evolution,'' \emph{IEEE Transactions
  on Communications}, vol.~57, no.~9, pp. 2534--2540, 2009.

\bibitem{Richardson-Feb-2001-A}
T.~J. Richardson, M.~A. Shokrollahi, and R.~L. Urbanke, ``Design of
  capacity-approaching irregular low-density parity-check codes,'' \emph{IEEE
  Transactions on Information Theory}, vol.~47, no.~2, 2001.

\bibitem{Xiao-Yu-Jun-2002-C}
H.~Xiao-Yu, E.~Eleftheriou, and D.~M. Arnold, ``Irregular progressive
  edge-growth ({PEG}) {T}anner graphs,'' in \emph{IEEE International Symposium
  on Information Theory}, Lausanne, Switzerland, 2002.

\end{thebibliography}


%%%%
%%%% Author Bios
%%%%

\begin{IEEEbiography}[{\includegraphics[width=1in,height=1.25in,clip,keepaspectratio,bb= 0 0 462 538]{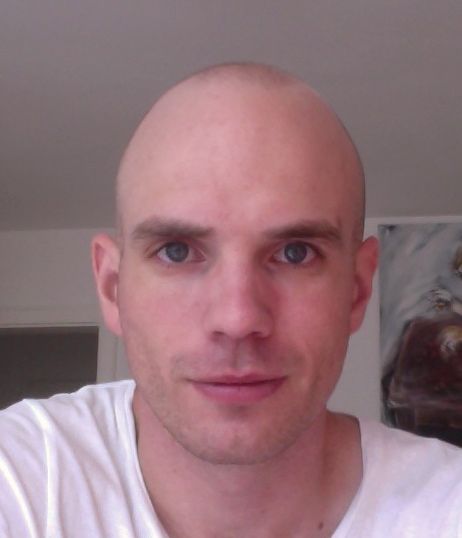}}]
{Roy Timo} (S'06-M'09) is a Research Fellow with the Institute for Telecommunications Research at the University of South Australia. Dr. Timo received the Bachelor of Engineering (Hons.) and Ph.D. degrees from The Australian National University in July 2005 and December 2009 respectively; he was a NICTA-enhanced Ph.D. candidate at NICTA's Canberra Research Laboratory. He held a Visiting Postdoctoral Research Associate position with the Department of Electrical Engineering at Princeton University in 2011 and 2012. He is a member of IEEE and the IEEE Information Theory Society.

Dr. Timo's research interests primarily lie within the fields of information theory and ergodic theory; in particular, the Shannon limits of source coding, channel coding and joint source-channel coding in networks.  
\end{IEEEbiography}

\begin{IEEEbiography}[{\includegraphics[width=1in,height=1.25in,clip,keepaspectratio,bb = 0 0 1153 1437]{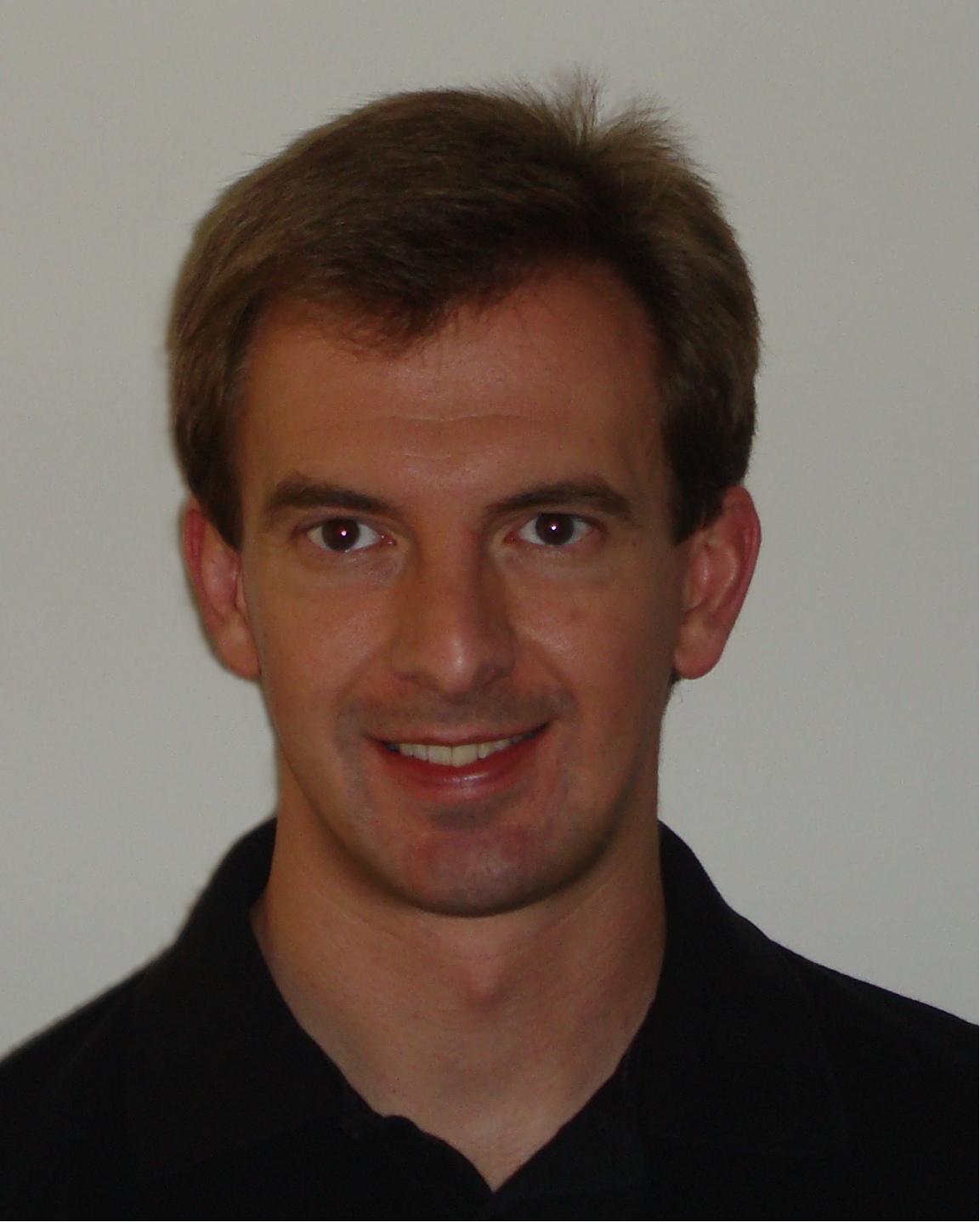}}]
{Gottfried Lechner} (S'03, M'08) was born on August 1, 1975 in Vienna, Austria. He received his Dipl.-Ing. and Dr.techn. degrees from the Vienna University of Technology (Vienna, Austria) in 2003 and 2007, respectively. From 2002 to 2008, he was a Researcher at the Telecommunications Research Centre Vienna (ftw) in the area of Signal- and Information Processing. Since 2008 he is a Research Fellow at the Institute for Telecommunications Research (ITR) at the University of South Australia (Adelaide, Australia). He is a Member of the IEEE and a member of the IEEE Information Theory and Communications Societies. His research interests include sparse graph codes, iterative techniques, source- and channel coding, cooperative communications, and optical communications.
\end{IEEEbiography}

\begin{IEEEbiography}[{\includegraphics[width=1in,height=1.25in,clip,keepaspectratio,bb = 0 0 580 725]{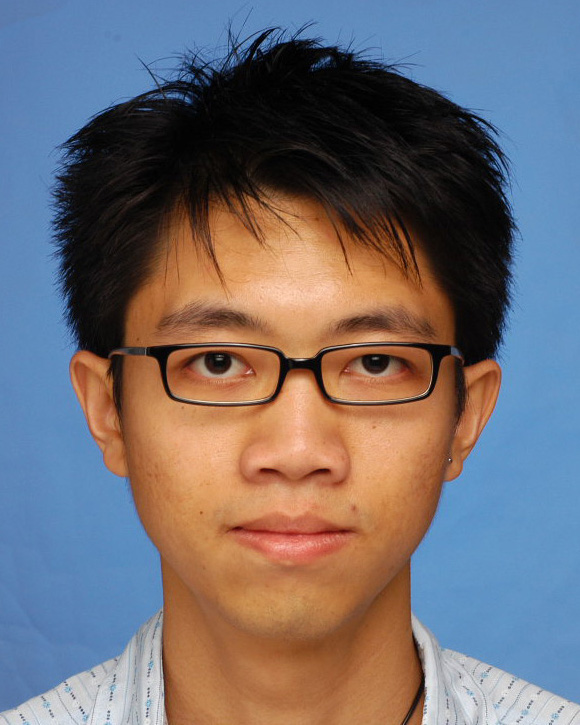}}]{Lawrence Ong} (S'05--M'10) received the BEng (1st Hons) degree in electrical engineering from the National University of Singapore (NUS), Singapore, in 2001. He subsequently received the MPhil degree from the University of Cambridge, UK, in 2004 and the PhD degree from NUS in 2008.

He was with MobileOne, Singapore, as a system engineer from 2001 to 2002. He was a research fellow at NUS, from 2007 to 2008. From 2008 to 2012, he was a postdoctoral researcher at The University of Newcastle, Australia. In 2012, he was awarded the Discovery Early Career Researcher Award (DECRA) by the Australian Research Council (ARC). He is currently a DECRA fellow at The University of Newcastle.
\end{IEEEbiography}

\begin{IEEEbiography}[{\includegraphics[width=1in,height=1.25in,clip,keepaspectratio,bb = 0 0 1022 1277]{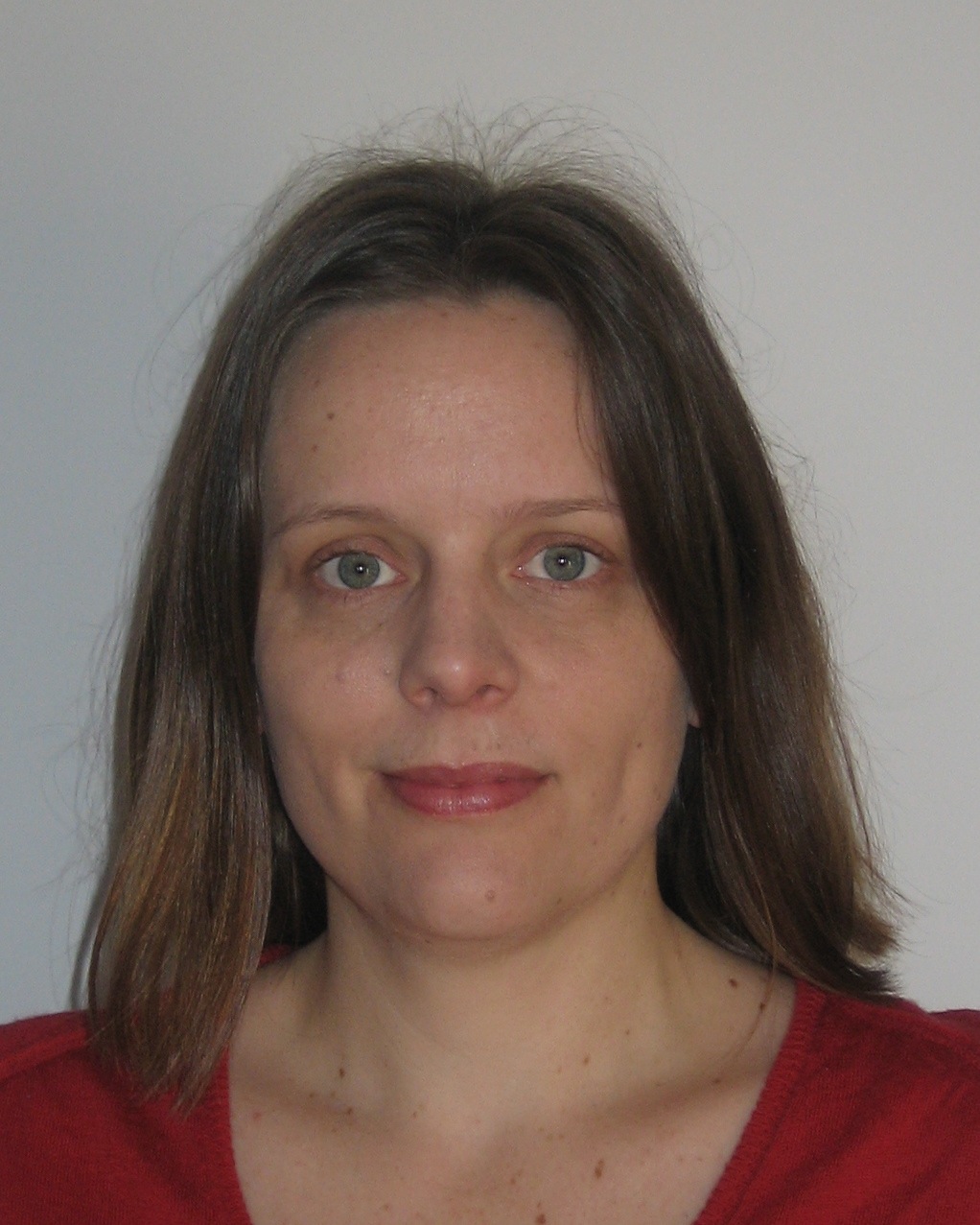}}]{Sarah Johnson} received the B.E. (Hons) degree in electrical
engineering in 2000, and PhD in 2004, both from the University of
Newcastle, Australia. She then held a postdoctoral position with the
Wireless Signal Processing Program, National ICT Australia before
returning to the University of Newcastle where she is now an Australian Research Council Future Fellow. Sarah's research interests are in the field of
error correction and information theory, and in particular the area
of codes for iterative decoding. She is the author of a book on
iterative error correction published by Cambridge University Press.
\end{IEEEbiography}

\end{document}